\documentclass[12pt, letterpaper]{article}
\usepackage{microtype}
\usepackage[T1]{fontenc}
\usepackage[margin=1in]{geometry}

\usepackage{graphicx, float, caption}
\usepackage{comment}
\usepackage{setspace}
\usepackage{appendix}
\usepackage{subcaption}
\usepackage{amsmath,amssymb,mathrsfs,amsthm}

\usepackage[colorlinks=true,linkcolor=blue,citecolor=blue]{hyperref}
\usepackage{natbib}
\usepackage{enumitem}

\usepackage{fnpct}

\usepackage{times}

\DeclareMathOperator{\E}{\mathbf{E}}

\newtheorem{proposition}{Proposition}
\newtheorem{corollary}{Corollary}[proposition]
\newtheorem{lemma}{Lemma}

\theoremstyle{definition}

\newtheorem{example}{Example}
\newtheorem{remark}{Remark}
\newtheorem*{example*}{Running Example}

\newcommand\cites[1]{\citeauthor{#1}'s\ (\citeyear{#1})}

\title{Tournaments with Managerial Discretion
\thanks{The authors are grateful to Chris Avery, Pablo Casas-Arce, Krishna Dasaratha, Maciej Kotowski, Jonathan Libgober, Bart Lipman, Chiara Margaria, Jawwad Noor, V.~G.~Narayanan, Juan Ortner, Krishna Palepu, Dennis Yao, and participants at Harvard University lunch seminars and the 36th Stony Brook Game Theory Conference for their helpful comments and discussion. %
}
}

\author{
 \textbf{Peiran Xiao}\thanks{Department of Economics, University of Southern California. Email: peiran.xiao@usc.edu.} \\
 University of Southern California
\and 
\textbf{Hashim Zaman}\thanks{Harvard Business School, Harvard University. Email: hzaman@hbs.edu.}\\
Harvard University
}

\date{\today}

\begin{document}
\maketitle

\onehalfspacing

\begin{abstract}
    We study tournaments with managerial discretion in hiring. A manager selects a coworker from a pool of candidates and then competes against him in a Lazear--Rosen--style tournament with a prize equal to a share of total output. A profit-maximizing principal sets the prize share together with a head start (or handicap)---an advantage (or disadvantage) in the output comparison---granted to the manager.  The head start affects output through three channels: (i) encouraging the manager, (ii) discouraging the new hire, and (iii) inducing the manager to hire a stronger candidate. The hiring effect dominates the discouragement effect until the strongest candidate is hired; beyond that point, any further head start discourages the new hire more than it encourages the manager. The optimal contract therefore grants a head start just large enough to induce the manager to hire the strongest candidate.
\end{abstract}


\clearpage

\begin{quote}
\emph{``[W]hen compensation is relative, and when the individuals who do the hiring are
to be in the same pool with those hired, there is an incentive to hire people
strategically. Incumbents do not want competition from good outsiders,
and so they tend to hire lower-quality people than would otherwise be
optimal for the firm.''}

  \hfill   ---Edward P. Lazear, \emph{Personnel Economics}
\end{quote}

\section{Introduction}

In hierarchical organizations, principals often delegate hiring and promotion decisions to lower-level managers, who are better positioned to evaluate candidates' abilities \citep{Sengupta2004,Haegele2026}.
For example, direct contact with candidates and familiarity with local conditions may give managers soft information about candidates' job-relevant skills and fit with the team \citep{DellerSandino2020,Frankel2021,WuLiu2026}, and such information may be difficult to communicate to the principal \citep{Dessein2002}.
\footnote{
 Decentralized hiring by local business-unit managers, commonly referred to as local manager hiring, is prevalent \citep{ChenJungbauerWang2023,WuLiu2026}.
  For our purposes, hiring and promotion are equivalent whenever the manager selects a candidate who subsequently joins the manager's tournament pool.
  Throughout, we use ``she'' for the manager and ``he'' for the selected candidate.
}
Under tournament-based incentive schemes, however, the selected candidate becomes both a coworker and a competitor of the incumbent manager.
This creates an incentive for the manager to sabotage the selection process by choosing a lower-ability candidate than would otherwise be optimal for the firm in order to prevent competition \citep{carmichael1988incentives,lazear1995personnel,FriebelRaith2004}.
For example, in academia, universities face the challenge of ensuring that incumbent faculty members hire the best possible candidates, and tenure is considered essential for this purpose.
In politics, incumbents tend to appoint weaker deputies to prevent future competition.
Similar dynamics also arise in corporate environments.
As Steve Jobs famously remarked, ``A players hire A players, B players hire C players, and C players hire D players.''
\footnote{
Obtained from \cite{Kawasaki2015}.
Moreover, \cite{Sullivan2011}  notes that weak managers often ``don’t even try to hire superior talent'' because they fear being overshadowed or displaced; instead, they hire employees with lower ability to increase their job security.  
}
In a survey of 336 corporate executives in the U.S. across various industries, \cite{ZamanLakhani2024} asked if respondents have ``ever observed a colleague disapprove hiring of a high-ability candidate to avoid potential competition for himself or herself.'' Among those whose firms operated on relative performance evaluation, over 30\% answered in the affirmative.

Despite the observational evidence, 
tournaments with managerial discretion and sabotage in the hiring process remain underexplored.
To fill this gap, we model a two-player tournament in which the manager has discretion over hiring a co-worker, who subsequently becomes a competitor. The manager chooses the new hire’s ability and then competes against him in a Lazear-Rosen-style tournament \citep{lazear1981rank}, where the winner receives a share of total output.
For concreteness, our primary example is a technology company in which a senior employee is delegated authority to hire a junior colleague for her team.
The new hire contributes to team output but also becomes a potential rival of the incumbent in an internal tournament for bonuses, promotions, or other rewards tied to team performance.

We propose that the principal can offer the manager a \emph{head start} to partially insulate her from potential competition. The head start is an advantage granted to the manager when comparing her output with the new hire's, which can be interpreted as a premium for seniority within the organization (see, e.g., \cite{Konrad2002, Siegel2014}). 
For example, in competition for bonuses or promotions, firms can provide incumbency advantages to employees who have been at the firm longer.
In organizations, incumbents can receive more credit than they deserve because of seniority.
Head starts, which are equivalent to handicapping opponents, are commonly used in tournaments \emph{without} managerial discretion---such as organizational and sports contests---to provide incentives (\cite{lazear1981rank, o1984economic, DrugovRyvkin2017, fu2020optimal}).
By contrast, we explore the role of head starts in mitigating managerial sabotage in hiring, in addition to providing incentives in tournaments.

The principal sets the tournament prize as a constant share of total output and can commit to this share through equity-based compensation contracts or reputational concerns.
In our technology-company example, the better-performing employee may receive a bonus or earn a promotion accompanied by an equity award, such as firm stock, whose value is tied to team or firm performance.
This incentive structure makes the new hire both a collaborator and a competitor: hiring a lower-ability agent increases the manager's probability of winning but reduces the size of the tournament prize.
Incentive schemes based on jointly achieved outcomes are common in organizations \citep{Holmstrom1982, DaiToikka2022, DasarathaGolubShah2024}.
Output-dependent prizes are also used in tournaments to avoid uncertainty, reduce wage costs, and discourage peer-to-peer sabotage and collusion \citep{GuthLevinskyPull2016, DanilovHarbringIrlenbusch2019, GloklerPullStadler2022}.
\footnote{
Investment banks provide another illustration of compensation tied to firm performance.
For example, Goldman Sachs allocated 34\% of its net revenue to its compensation pool in 2023 \citep{Goldman2023Q2}.
}

We find that the head start has varying effects on the efforts of both the manager and the new hire. First, it biases the tournament in favor of the manager and incentivizes her to exert effort, thereby having an \emph{encouragement effect}. However, its impact on the new hire's effort is rather mixed. On the one hand, the bias discourages the new hire from investing effort, thereby having a \emph{discouragement effect}. On the other hand, it mitigates the competition faced by the manager and leads her to hire a higher-ability agent, which increases the new hire’s effort relative to the situation where a lower-ability agent would have been hired. This positive \emph{hiring effect} dominates the {discouragement effect} in equilibrium until the highest-ability agent is hired. 
Once the highest-ability candidate is hired, any further head start discourages the new hire more than it encourages the manager, thereby decreasing total output.
Therefore, the optimal head start is just enough to induce the manager to hire the strongest candidate.

The payout ratio (i.e., the share of total output awarded to the winner) also has mixed effects on the principal's profit. 
While a higher payout ratio reduces the principal's share of total output, it affects the agents' effort in two ways.
On the one hand, a higher payout ratio incentivizes both the manager and the new hire to exert more effort. 
On the other hand, it affects the manager's choice of the new hire's ability: hiring a higher-ability agent increases the size of the tournament prize but lowers her chances of winning.
We characterize the profit-maximizing head start and payout ratio for the principal to take advantage of tournament incentives while mitigating hiring sabotage.
\footnote{
In Appendix~\ref{app:beta}, we explore an extension where the loser also receives a constant share of total output. We show that our main results remain robust: the optimal contract assigns a zero share to the loser, as giving a head start to the manager is more effective at mitigating hiring sabotage.
}

In this paper, we explore the unique dynamics in tournaments that involve managerial discretion in hiring the other player.
To the best of our knowledge, we are the first in contest theory to study the role of a player's discretion in choosing her competitor, which leads to sabotage during this process.
Moreover, we study the use of head starts and output-dependent prizes in tournaments to mitigate hiring sabotage. While the literature has explored the use of head starts (or handicaps) to provide incentives in the tournament stage (e.g., \cite{lazear1981rank,o1984economic, DrugovRyvkin2017}) and tournaments with output-dependent prizes (e.g., \cite{Chung1996, BayeHoppe2003, GershkovLiSchweinzer2009,GuthLevinskyPull2016, DanilovHarbringIrlenbusch2019}), we are the first to explore their roles in mitigating managerial sabotage in the hiring process. 
Additionally, even in the absence of managerial discretion, we show that giving a head start to the lower-ability agent can decrease total effort in tournaments with output-dependent prizes, contrary to the common wisdom of ``leveling the playing field.''
Our results on the effects of the head start, noise, and abilities on individual efforts are also novel in the setting with output-dependent prizes.

The tournament literature has explored peer-to-peer sabotage (e.g., \cite{lazear1989pay, skaperdas1995modeling, chen2003sabotage, krakel2005helping, munster2007selection,HarbringIrlenbusch2008, gurtler2010sabotage, deutscher2013sabotage}).%
\footnote{
See \cite{chowdhury2015sabotage} for a survey of (peer-to-peer) sabotage in tournaments.  
}
The literature has also found that output-dependent prizes can mitigate peer-to-peer sabotage and encourage helping \citep{DanilovHarbringIrlenbusch2019, GloklerPullStadler2022} and that handicaps can exacerbate peer-to-peer sabotage \citep{brown2017hidden}.
By contrast, we consider a setting in which the manager has discretion over hiring a new employee and may sabotage the \emph{hiring} process to prevent competition from the new hire. 
While managerial authority in hiring and promotion has been explored in organizational economics (e.g., \cite{carmichael1988incentives, lazear1995personnel, Dessein2002, FriebelRaith2004, Sengupta2004, Chen2024, Haegele2026}), its implications for tournaments remain unexplored.

The literature has considered the use of head starts or handicaps in tournaments to mitigate heterogeneity in players' abilities (e.g., \cite{lazear1981rank, o1984economic}). The common wisdom is to give a head start to the weaker player (i.e., handicap the stronger player) to ``level the playing field,'' as player heterogeneity would diminish incentives for both players.
Recent papers have challenged this common wisdom. \cite{DrugovRyvkin2017} show that head starts or handicaps can be optimal even when players have the same ability.
\cite{fu2020optimal} and \cite{drugov2022hunting} find that giving a head start to the stronger player can increase total effort.
Consistent with these findings, we find that in Lazear-Rosen-style tournaments with output-dependent prizes, a head start to the stronger player can increase total effort.
To the best of our knowledge, the literature has not explored tournaments where an agent has the discretion to hire another, \emph{a fortiori} the use of head starts to mitigate sabotage in this process.%

\section{The Model}

\subsection{Setup}
Consider a two-player tournament in which a player, the manager $m$, has the discretion to hire the other agent $n$.
This framework captures situations where an organization operating under tournament incentives needs to hire a new agent, but the principal cannot observe the abilities of the candidates and must delegate the hiring decision to a manager.
The manager of ability (type) $\theta_m\in  \Theta \equiv [\underline \theta, \bar \theta]$ observes the abilities of the candidates and selects a new hire of ability $\theta_n\in \Theta$.
We assume that $\underline\theta>0$ is sufficiently small.
Once the hiring decision is made, the manager competes with the new hire in a Lazear--Rosen--style tournament.
In the tournament, both agents have common knowledge of their abilities $(\theta_m,\theta_n)$, and agent $i\in\{m,n\}$ invests effort $e_i\geq0$ towards production at a cost $c(e_i)/\theta_i$. Assume that $c(\cdot)$ is twice continuously differentiable, strictly increasing, and $c''(e)>0$ for all $e\geq0$, with $c(0)=c'(0)=0$ and $\lim_{e\to\infty} c'(e)>2\bar\theta$. 
The output of agent $i$ is given by $y_i = e_i+\varepsilon_i$, where the term $\varepsilon_i$ captures a zero-mean random shock, which can also be interpreted as noise in performance evaluation.
We assume that the tournament prize $V(y)$ depends on the total output $y=y_m+y_n$, instead of a fixed prize, in order to provide incentives for the manager to hire a high-ability agent.%
\footnote{
    If the tournament prize does not depend on total output, the manager will always hire the lowest-ability agent. 
}
For tractability, we further assume the tournament prize is a constant share of total output, i.e., $V(y) = \alpha y$, which can be interpreted as a performance bonus or equity award tied to team or firm output. %
\footnote{ 
    See \citet{GershkovLiSchweinzer2009} and \cite{GuthLevinskyPull2016} for prizes as a constant share of total output in tournaments.
    In a Lazear-Rosen-style tournament without managerial discretion, \cite{GuthLevinskyPull2016} show that this prize structure is more cost-effective than a fixed prize. 
    Managerial discretion makes the output-dependent prize even more desirable because it incentivizes the manager to hire a high-ability agent.
    }
The profit-maximizing principal can commit to the payout ratio $\alpha\in(0,1)$ through formal equity compensation contracts or through repeated interaction and reputation.
In addition, she can also offer a \emph{head start} $h\in[-\bar{H}, \bar{H}]$ to the manager---an advantage in output comparison---to incentivize her to hire a higher-ability agent.
\footnote{
We allow for the possibility that $h<0$, which makes it a \emph{handicap}---a bias \emph{against} the manager when comparing outputs. 
In the absence of the manager’s discretion, it is optimal to give a head start to the higher-ability agent in our setting (see Remark~\ref{rmkl}).
}
In practice, the firm may reward better-performing agents with stock or promotion and offer a seniority premium in performance evaluations.
We assume that the bound $\bar{H}>0$ on the head start is large.
Given the head start $h$, the manager wins the prize if and only if $y_m + h \geq y_n$.

Alternatively, the tournament scheme can be viewed as a sharing contract: the principal retains a fraction $1-\alpha$ of the total output, and the remaining fraction $\alpha$ constitutes a compensation pool $V=\alpha y$ shared by the agents.
Specifically, each agent’s percentage share of the compensation pool in the sharing contract equals their probability of winning in the tournament, which is \emph{endogenously} determined by their effort levels. A head start $h>0$ increases the manager's percentage share of the compensation pool and reduces the new hire's percentage share.

The timing of the game is as follows. First, the principal knows the manager's ability and commits to the contract $(\alpha,h)$.
Next, the manager chooses an agent $\theta_n\in\Theta$ from the candidate pool.
Then, the manager and the new hire choose their effort levels in the tournament, where abilities are common knowledge to them.
Finally, the output is realized, and agents are paid according to the contract $(\alpha,h)$.
The solution concept we use is the pure-strategy subgame-perfect Nash equilibrium (SPNE).

\subsection{Tournament Stage}
Given the head start $h$, the manager $m$ will outperform $n$ with probability 
\begin{equation*}
    \Pr(y_m + h\geq y_n) = \Pr(e_m-e_n+h\geq   \varepsilon_n-\varepsilon_m ). %
\end{equation*}
Assume that the difference in random shocks,
$\varepsilon_n-\varepsilon_m$, follows a distribution $G(\cdot)$ with a continuously differentiable density $g(\cdot)$ that is unimodal and symmetric around zero (i.e., $g(x)=g(-x)$).
For example, the distribution can be generated by independent and identically distributed (i.i.d.) random shocks with a unimodal distribution (\citet[Theorem 2.2]{Purkayastha1998}).
Therefore, the probability that $m$ will outperform $n$ is $\Pr(y_m + h\geq y_n) = G(e_m-e_n+h)$.
By symmetry, the probability that $n$ will outperform $m$ is $G(e_n-e_m-h) = 1-G(e_m-e_n+h)$.
We also assume $\E[\varepsilon_n+\varepsilon_m\mid \varepsilon_n-\varepsilon_m]=0$ so that the expected total output is the same regardless of the winner's identity, which is satisfied if the random shocks are i.i.d.~and symmetric around zero.

Additionally, we parameterize the distribution of the difference in random shocks by $G(x)= F(x/\sigma)$, where $F$ is a CDF with unit variance, and $\sigma>0$ is a scale parameter that measures the noise in performance evaluation (i.e., the dispersiveness and the variance of $G$). 
\footnote{
After the parameterization, the distribution $G$ has a larger variance if and only if the difference in random shocks is more dispersed (\citet[Proposition 2]{DrugovRyvkin2020}).
When random shocks are i.i.d.~and within a location--scale family, their difference is more dispersed if and only if the random shocks themselves are more dispersed.
More generally, when random shocks are i.i.d.~and log-concave, their difference is more dispersed if the shocks themselves are more dispersed (\citet[Theorem 3.B.9]{ShakedShanthikumar2007}).
}
This distribution can be generated by i.i.d.~random shocks that belong to a location-scale family with variance $\sigma^2/2$ and mean zero \citep{DrugovRyvkin2020, MorganTumlinsonVardy2022} as well as correlated random shocks within a location-scale family, as shown in the following examples.

\begin{example}[Normal Distribution]
    If $\varepsilon_m,\varepsilon_n \stackrel{\text{i.i.d.}}{\sim} \mathcal{N}(0,\sigma^2/2)$, then $\varepsilon_n-\varepsilon_m \sim \mathcal{N}(0,\sigma^2)$. %
\end{example}

\begin{example}[Logistic Distribution]
    If $\varepsilon_m,\varepsilon_n \stackrel{\text{i.i.d.}}{\sim} \text{Gumbel}(-\gamma s,s)$, which has mean zero, then $\varepsilon_n-\varepsilon_m \sim \text{Logistic}(0,s)$. %
    Thus, after log-transforming effort, the tournament is equivalent to a \cite{Tullock1980} contest with discriminatory power $r=1/s$ and a multiplicative head start. 
\end{example}

\begin{example}[Uniform Distribution]
    If $\varepsilon_m,\varepsilon_n\sim \operatorname{Unif}[-\frac{M}{4},\frac{M}{4}]$ and $\varepsilon_m=-\varepsilon_n$, 
    then $\varepsilon_n-\varepsilon_m \sim \operatorname{Unif}[-\frac{M}{2},\frac{M}{2}]$, and $\E[\varepsilon_n+\varepsilon_m\mid \varepsilon_n-\varepsilon_m]=0$ also holds.%
    \footnote{
        If $\varepsilon_m$ and $\varepsilon_n$ are i.i.d. uniform random variables, then
        $\varepsilon_n-\varepsilon_m$ follows a symmetric triangular distribution. 
    }
\end{example}

We use backward induction to derive the pure-strategy SPNE.
Given the new hire's ability $\theta_n$, in the subgame $\Gamma(\theta_n)$ (i.e., the tournament stage), the expected payoffs for $m$ and $n$ are
\footnote{Note that the subgame is not supermodular in $(e_m,e_n)$ unless $g'(x)=0$ (uniform distribution).}
\begin{align}
   u_m(e_m,e_n,\theta_m) & = \alpha \cdot G(e_m-e_n+h)  (e_n+e_m) - \frac{c(e_m)}{\theta_m}, \\
   u_n(e_m,e_n,\theta_n) & = \alpha \cdot G(e_n-e_m-h)  (e_n+e_m) - \frac{c(e_n)}{\theta_n}.
\end{align}
Both agents will choose the effort levels $(e_m,e_n)$ to maximize their expected payoffs.
Since each agent can guarantee a nonnegative payoff by choosing zero effort,
effort levels with average cost $c(e_i)/e_i>2\bar\theta$ cannot arise in
equilibrium. Consequently, the equilibrium effort levels are bounded
(see Lemma~\ref{lem:bounded}).
The first-order conditions yield
\begin{align} 
  \frac{\partial u_m }{\partial  e_m}  & =  \alpha[(e_n+e_m)\cdot g(e_m-e_n+h) + G(e_m-e_n+h)] - \frac{c'(e_m)}{\theta_m} = 0, \label{FOC1}\\
  \frac{\partial u_n }{\partial e_n}  & = \alpha[(e_n+e_m)\cdot g(e_m-e_n+h) + G(e_n-e_m-h)] - \frac{c'(e_n)}{\theta_n}=0. \label{FOC2}
\end{align}
Following the standard practice in the tournament literature, we assume the variance of $G$ is sufficiently large  
so that the second-order conditions are satisfied, and a unique pure-strategy equilibrium exists. 
\footnote{%
    See \citet{lazear1981rank}, \citet{nalebuff1983prizes}, and \cite{RyvkinDrugov2020}.
    When $x$ is bounded and $g(x)$ is continuously differentiable, a large variance implies that $g(x)$ and $g'(x)/g(x)$ are sufficiently small (see Lemma~\ref{lem:prelim}).
    For example, the normal distribution $\mathcal{N}(0,\sigma^2)$ has $g(x)\leq \frac{1}{\sqrt{2\pi}\sigma}$ and $\frac{g'(x)}{g(x)} = -\frac{x}{\sigma^2}$.
    An alternative assumption is that $c''(e)$ is uniformly sufficiently large and $G$ is log-concave (see Remark~\ref{rem:weaker-assumptions}).
 }
Thus, the two equations above determine the equilibrium effort levels $(e_m(\theta_n),e_n(\theta_n))$ in the subgame as functions of the hiring decision $\theta_n$.
The following lemma provides the comparative statics.

\begin{lemma}
\label{lemma:1}
 Given $\theta_n$, the equilibrium effort levels in the subgame $\Gamma(\theta_n)$ are given by equations~\eqref{FOC1} and \eqref{FOC2}.
 The following comparative statics hold.
    \begin{enumerate}[label=(\roman*)]
        \item The manager's effort $e_m(\theta_n)$ is
        \begin{itemize}
            \item strictly increasing in the head start $h$, the payout ratio $\alpha$, and her ability $\theta_m$;
            \item decreasing in the noise in performance evaluation $\sigma$ if $h\geq 0$;%
            \footnote{
            The necessary and sufficient condition is that $h>-2e_m(\theta_n)$, i.e., the head start is not too unfavorable to her.
            Symmetrically, the new hire's effort is decreasing in $\sigma$ if and only if $h<2e_n(\theta_n)$.
            }
            
            \item increasing (decreasing) in the new hire's ability $\theta_n$ if $\theta_m\geq\theta_n$ and $h\geq0$  
            ($\theta_n\geq\theta_m$ and $h\leq 0$).

        \end{itemize}
        \item The new hire's effort $e_n(\theta_n)$ is
        \begin{itemize}
            \item strictly decreasing in the head start $h$, and strictly increasing in $\alpha$ and $\theta_n$;
            \item decreasing in $\sigma$ if $h\leq0$;%
            \item increasing (decreasing) in $\theta_m$ if $\theta_n\geq\theta_m$ and $h\leq0$  ($\theta_m\geq\theta_n$ and $h\geq 0$).
        \end{itemize}     
        \item There exist $\bar h_1,\bar h_2\geq 0$ such that if $h>\bar h_1$ ($h<-\bar h_2$), then $e_m(\theta_n)$ is increasing (decreasing) in $\theta_n$, and $e_n(\theta_n)$ is decreasing (increasing) in $\theta_m$.
        \footnote{
            The more precise condition is that the manager is more (less) likely to win; the lemma provides two sufficient conditions.
            For example, $\frac{de_m}{d\theta_n}$ is positive (negative) if $g'(e_m^*-e_n^*+h)$ is negative (positive), which holds when the manager is more (less) likely to win because $g(x)$ is unimodal.
            For strictly unimodal $g(x)$, the condition is necessary and sufficient.
            For the uniform distribution, $\frac{de_m}{d\theta_n}=0$ because $g'(x)\equiv0$ on the support.
            }
     \end{enumerate}
     
\end{lemma}

\begin{remark}
    Because the abilities of both agents are fixed in this lemma, the results extend the literature on head starts (or handicaps) in tournaments with fixed prizes to a setting with output-dependent prizes. 
    The effects of the head start and noise are consistent with 
    this literature (Cf.~\cite{o1984economic, DrugovRyvkin2020}). 
    To our knowledge, the comparative statics of individual efforts with respect to the head start, noise, and abilities, which may be interdependent, are novel in this setting with output-dependent prizes.
\end{remark}

The proof is in Appendix~\ref{app:proofs}. Given the new hire’s ability $\theta_n$, a head start $h$ increases the manager’s chance of winning, encouraging the manager to invest more effort while discouraging the new hire from doing so. We will explore these two effects in detail in the next section, where we also show that, because of the output-dependent prizes, it can be suboptimal to give a head start to the lower-ability agent (even in the absence of managerial discretion), in contrast to the common wisdom of ``leveling the playing field'' (see Lemma~\ref{lemma:3}).

As the payout ratio $\alpha$ increases, the stakes in the tournament rise, providing stronger incentives for both the manager and the new hire to invest effort. Thus, both agents invest more effort as the payout ratio increases.

Within a location-scale family of the distribution $G$, as the scale parameter $\sigma$ (or the variance $\sigma^2$) increases, the difference in random shocks becomes more dispersed, and the performance evaluation becomes noisier.
Therefore, the manager and the new hire have weaker incentives to invest effort when the competition is fair (i.e., $h=0$) or the head start is in their favor. 
However, the effect on effort is reversed if the head start is sufficiently unfavorable to them.
For example, when the head start $h>0$ is sufficiently large, a noisier evaluation encourages the new hire to invest more effort, since it gives him some hope of winning despite the large bias against him.
 
Finally, as the manager’s ability $\theta_m$ increases, she invests more effort due to a lower marginal cost of effort. At the same time, the new hire is incentivized to invest more effort to outperform the manager if he is more likely to win than the manager.
Otherwise, if the new hire is less likely to win, he will invest less effort when the manager's ability increases.
Technically, this arises because the new hire's payoff is supermodular (submodular) in his own effort and the manager's effort when the new hire is more (less) likely to win.
By symmetry, as the new hire's ability $\theta_n$ increases, he invests more effort, whereas the manager invests more (less) effort if she is more (less) likely to win.
Intuitively, the manager is more (less) likely to win when her ability is high (low) compared to the new hire's or when the head start is large (small).

The following example illustrates these results under the uniform distribution and quadratic costs.

\begin{example*}[Uniform--Quadratic]
        Suppose $G \sim \operatorname{Unif}[-\frac{M}{2},\frac{M}{2}]$, $c(e)=e^2/2$, and $M$ is sufficiently large that all winning probabilities are interior. The equilibrium effort levels in the subgame $\Gamma(\theta_n)$ are given by
    \begin{align*} 
        e_m(\theta_n) &=  \frac{M/2+ h}{M/\alpha\theta_m -2},\\
        e_n(\theta_n) &=   \frac{M/2-h}{M/\alpha \theta_n -2}.
    \end{align*}
    Under the uniform distribution, an agent’s marginal benefit from increasing effort is independent of the opponent’s effort level.
    Consequently, strategic interdependence in effort choices disappears, and agents' equilibrium effort levels do not depend on their opponents’ abilities.%
    \footnote{
        With a uniform distribution, it is no longer necessary to assume that abilities are common knowledge to both agents in the tournament, allowing the possibility that the new hire does not know the manager's ability. 
        }
    The effort levels are increasing in the amount of head start in their favor, the payout ratio, and their own abilities.
    Moreover, both the manager's and the new hire's efforts decrease in the variance of $G$ if and only if the head start is not too unfavorable to them---i.e., $h\geq -\alpha \theta_m$ and $h\leq \alpha \theta_n$ respectively.
\end{example*}

\subsection{Hiring Stage}
At the hiring stage, the manager chooses $\theta_n\in\Theta$ to maximize her expected payoff in the subgame $\Gamma(\theta_n)$.
By Lemma~\ref{lemma:1}, because both $c(\cdot)$ and $g(\cdot)$ are continuously differentiable, $e_n(\theta_n)$ is continuously differentiable and increasing in $\theta_n$.
Therefore, we have
\begin{align} \label{FOC3}
   \frac{d  u_m (e_m(\theta_n),e_n(\theta_n))}{d \theta_n} =  
   \frac{\partial u_m}{\partial e_n} e_n'(\theta_n) =
   \alpha   [G(e_m-e_n+h) - (e_m+e_n)\cdot g(e_m-e_n+h)]e_n'(\theta_n).%
\end{align}
When the manager hires a higher-ability agent, the agent will invest more effort at the tournament stage, which increases the size of the tournament prize but also reduces the manager's probability of winning.
The new hire's ability affects the manager's payoff only through his own effort ($e_n$), and not through the manager's effort, as $e_m$ is chosen optimally by the manager at the tournament stage.
Therefore, the manager chooses $\theta_n$ to determine the new hire’s effort $e_n(\theta_n)$ in the subgame and maximize her expected payoff.

Hence, in the SPNE, when the hiring decision $\theta_n^*\in(\underline{\theta},\bar\theta)$, the first-order conditions for $\theta_n^*$ and the equilibrium effort levels $(e_m^*,e_n^*)$ are given by 
\begin{align}
&G(e_m^*-e_n^*+h) = (e_m^*+e_n^*)\cdot g(e_m^*-e_n^*+h) \label{eqm1} \\
&c'(e_n^* )/\theta_n^* =\alpha, \label{eqm2} \\
&c'(e_m^* )/\theta_m  = 2\alpha  \cdot G(e_m^* -e_n^* +h).\label{eqm3}
\end{align}
The optimal hiring decision $\theta_n^*$ balances the marginal benefit of hiring a higher-ability agent (i.e., a higher tournament prize) and the marginal cost (i.e., a lower probability of winning), as determined by equation~\eqref{eqm1}.
Given his ability $\theta_n^*$ chosen optimally by the manager, the first-order condition for the new hire's effort collapses to equation~\eqref{eqm2}:  in equilibrium, his marginal cost of effort equals its marginal effect on the expected tournament prize (i.e., $\partial \E[V]/\partial e_n = \alpha$).
On the other hand, the manager's marginal cost of effort is higher (lower) than its marginal effect on the expected prize if she is more (less) likely to win the tournament.

In general, it is not possible to obtain a closed-form solution.
When the hiring decision $\theta_n^*$ takes the corner solution $\theta_n^*\in\{\underline\theta,\bar\theta\}$, the comparative statics in Lemma~\ref{lemma:1} apply.
For an interior solution, we obtain the following results.

\begin{lemma}%
\label{lemma:2}
     Assume $\theta_n^*\in(\underline{\theta},\bar\theta)$.
     The equilibrium effort levels $(e_m^*,e_n^*)$ and hiring decision $\theta_n^*$ are given by equations~\eqref{eqm1}--\eqref{eqm3}.
     The following comparative statics hold.
    \begin{enumerate} [label=(\roman*)]
        \item The manager's effort $e_m^*$ is
        \begin{itemize}
            \item strictly increasing in the head start $h$, the payout ratio $\alpha$, and her ability $\theta_m$;
            \item strictly decreasing in the noise in performance evaluation $\sigma$ if $h\geq0$.
            \footnote{
            The necessary and sufficient condition is that the head start is not too unfavorable to the manager, i.e., $h>-2e_m^*$.
            } 
        \end{itemize}
        
        \item
        The hiring decision $\theta_n^*$ is strictly increasing in $h$ and $\sigma$, and strictly decreasing in $\alpha$.
        
        \item 
        The new hire's effort $e_n^*$ is strictly increasing in $h$ and $\sigma$.

        \item There exist $\bar h_1,\bar h_2>0$ such that if $h>\bar h_1$ ($h<-\bar h_2$), then 
        \begin{itemize}
        \item $\theta_n^*$ is increasing (decreasing) in $\theta_m$;
        \item $e_n^*$ is increasing (decreasing) in $\theta_m$ and $\alpha$.
        \footnote{The more precise condition for each statement in (iv) and (v) is that the manager is more (less) likely to win, which is necessary and sufficient for strictly unimodal $g(x)$.}
        \end{itemize}

        \item There exist $\bar\theta_m, \underline\theta_m\in \Theta$ such that if $\theta_m>\bar\theta_m$ ($\theta_m<\underline\theta_m$) and $h\geq0$ ($h\leq0$), then
        \begin{itemize}
        \item $\theta_n^*$ is increasing (decreasing) in $\theta_m$;
        \item $e_n^*$ is increasing (decreasing) in $\theta_m$ and $\alpha$.
        \end{itemize}
    \end{enumerate}
    
\end{lemma}

A higher head start encourages the manager to invest more effort. Additionally, increasing the head start further insulates her from competition, leading her to hire a higher-ability agent until the strongest candidate is hired.
The head start affects the new hire’s effort through two channels: on the one hand, it discourages effort by lowering his chance of winning; on the other, it leads to a higher-ability hire, which in turn increases his effort. Overall, the net effect on the new hire's effort is positive. We will discuss them in detail in the next section.

As the difference in random shocks becomes more dispersed (captured by a higher $\sigma$), performance evaluation becomes noisier, so the manager has weaker incentives to invest effort when the head start favors her ($h\geq0$). 
A noisier performance evaluation also makes the manager less concerned about competing with a higher-ability agent, leading her to hire a higher-ability candidate. The variance affects the new hire’s effort through two channels: on the one hand, it discourages effort by making performance evaluation noisier; on the other, it results in a higher-ability new hire, which in turn increases his effort. Overall, the net effect on the new hire’s effort is positive.

As the payout ratio  $\alpha$ increases, the stakes in the tournament rise, providing stronger incentives for the manager to invest effort. The payout ratio also affects the new hire’s effort through two channels. On the one hand, it directly increases the new hire’s incentives to invest effort.
On the other hand, it increases the manager’s incentive to prevent competition by hiring a lower-ability agent, especially when she is less likely to win.
Overall, the net effect on the new hire’s effort is positive (negative) if the manager is more (less) likely to win.
The manager is more (less) likely to win if her ability is sufficiently high (low) and the head start favors (disfavors) her.

Similarly, as the manager’s ability $\theta_m$ increases, she will invest more effort because her marginal cost of effort is lower. The manager’s ability also affects the new hire’s effort through two channels: if the manager is more likely to win, it incentivizes the new hire to exert more effort to outperform the manager; on the other hand, a higher-ability manager chooses a higher-ability agent, which further increases the new hire’s effort. Overall, the net effect on the new hire’s effort is positive (negative) if the manager is more (less) likely to win.

The following example illustrates the results with a uniform distribution and a quadratic cost function.

\begin{example*}[Uniform--Quadratic]
    Suppose $G \sim \operatorname{Unif}[-\frac{M}{2},\frac{M}{2}]$ and $c(e)=e^2/2$, with $M$ sufficiently large so that all winning probabilities are interior. 
    The equilibrium hiring decision and effort levels are given by
    \begin{align*} 
       \theta_n^*  &= \min{\left(\frac{h+M/2}{2\alpha}, \bar \theta \right)}   \label{theta_n} ,\\
         e_m^*  &= \frac{M/2+ h  }{M/\alpha\theta_m -2},\\
          e^*_n &\equiv e_n (\theta_n^* )  =    \min{\left(\frac{h+M/2}{2},  \frac{M/2-h}{M/\alpha\bar \theta -2}   \right)}.
    \end{align*}
    The uniform assumption removes the strategic interdependence of effort choices in the tournament stage (because $g'(x)=0$), so neither the new hire's effort $e_n^*$ nor his ability $\theta_n^*$ depends on the manager's ability $\theta_m$.
    Furthermore, the new hire's effort $e_n^*$ is also independent of the payout ratio $\alpha$ when $\theta_n^*<\bar\theta$  because its positive effect on the new hire's effort is exactly offset by its negative effect through the decrease in the new hire's ability due to increased managerial sabotage.
\end{example*}

\section{Optimal Head Start}
The effects of the head start on the expected total output $\E[y^*]=e_m^*+e_n^*$ can be decomposed into three components:
\begin{equation*}
 \frac{d\E[y^*]}{dh} %
 =  {\frac{de_m^*}{dh}} +  \frac{\partial e_n (\theta_n,h) }{\partial h}\Big|_{\theta_n=\theta_n^*(h)} +   \frac{\partial e_n(\theta_n,h)}{\partial  \theta_n}\Big|_{\theta_n=\theta_n^*(h)}  \theta_n^{*\prime}(h).
\end{equation*}
According to Lemmas 1 and 2, we identify the three effects and their directions as follows: 
\begin{enumerate}
    \item Encouragement effect on $e_m$: $\frac{de_m^*}{dh}>0$.
    \item Discouragement effect on $e_n$: $\frac{\partial e_n (\theta_n^*(h),h) }{\partial h}<0$.
    \item Hiring effect on $e_n$ (through $\theta^*_n$): $\frac{\partial e_n(\theta_n^*(h),h)}{\partial  \theta_n} \theta_n^{*\prime}(h)\geq0$ ($>0$ if $\theta_n^*(h) < \bar\theta$).
\end{enumerate}

Intuitively, the head start increases the manager's probability of winning. Since the tournament prize is increasing in effort, it has an \emph{encouragement effect} on the manager by increasing her marginal return to effort.   
However, the impact of the head start on the new hire's effort is rather mixed. On the one hand, holding his ability $\theta_n$ fixed, the head start reduces his marginal return to effort, thereby having a \emph{discouragement effect}. On the other hand, the head start partially insulates the manager from competition and leads the manager to hire a higher-ability agent. This \emph{hiring effect} results in greater effort from the higher-ability new hire.

Importantly,
by Lemma~\ref{lemma:2}, the hiring effect dominates the discouragement effect until the highest-ability agent is hired (i.e., when $\theta_n^*(h)<\bar \theta$) because
\begin{equation*}
 \frac{de_n^*}{dh} = \underbrace{\frac{\partial e_n (\theta_n,h) }{\partial h}\Big|_{\theta_n=\theta_n^*(h)}}_{\text{Discouragement effect}} + 
 \underbrace{ \frac{\partial e_n(\theta_n,h)}{\partial  \theta_n}\Big|_{\theta_n=\theta_n^*(h)}  \theta_n^{*\prime}(h)}_{\text{Hiring effect}} >0
 \label{eq: DECE}
 \end{equation*}
Thus, the head start always increases total output until the strongest candidate is hired, so the optimal head start must (at least) induce the manager to hire the highest-ability agent $\bar\theta$.

\begin{example*}[Uniform--Quadratic]
    Assume $G \sim \operatorname{Unif}[-\frac{M}{2},\frac{M}{2}]$ and $c(e)=e^2/2$. 
    When $\theta_n^*(h) \in(\underline\theta, \bar\theta)$,
    the encouragement and discouragement effects are given by
    $$\frac{de_m^*}{dh} = \frac{1}{M/\alpha \theta_m-2}>0,\quad \frac{\partial e_n (\theta_n^*,h) }{\partial h} = - \frac{1}{M/\alpha \theta_n^* -2}<0$$
    respectively.
    The sum of the discouragement and the hiring effects is 
    $\frac{de_n^*}{dh}  = \frac{1}{2}$.
  \end{example*} 

Because $\theta_n^{*}(h)$ is strictly increasing in $h$ until the strongest candidate is hired (i.e., $\theta_n^*(h)=\bar\theta$), we denote by $\bar h(\alpha)$ the head start just enough to induce the manager to hire the highest-ability candidate---i.e., $\theta_n^*(h)=\bar \theta$ for all $h\geq\bar h(\alpha)$ and $\theta_n^*(h)<\bar \theta$ for all $h<\bar h(\alpha)$.
In general, by Lemma~\ref{lemma:2}, $\bar h(\alpha)$ is increasing in $\alpha$ and decreasing in the variance of $G$.
This is because a higher payout ratio $\alpha$ or a lower noise in the output evaluation increases the manager's incentives to sabotage, and therefore requires a higher head start to induce the manager to hire the strongest candidate.

\begin{example*}[Uniform--Quadratic]
    For $G \sim \operatorname{Unif}[-\frac{M}{2},\frac{M}{2}]$ and $c(e)=e^2/2$, 
    the head start just enough to induce the manager to hire the strongest candidate is 
    $\bar h(\alpha) = 2\alpha\bar\theta-M/2$.
  \end{example*}

Having shown that the optimal head start is at least \(\bar h(\alpha)\), we next ask whether it is desirable to give the manager a further head start.
As $\theta_n$ cannot increase beyond $\bar\theta$, a further head start will no longer have the hiring effect.
Hence, we need to compare the discouragement effect (on the new hire) and the encouragement effect (on the manager) in the absence of the hiring effect (i.e., holding $\theta_n$ fixed).
The following lemma provides a necessary and sufficient condition under which the discouragement effect dominates the encouragement effect.
\begin{lemma} \label{lemma:3}
   Given $\theta_n$, the expected total output ($e_m^*+e_n^*$) is decreasing in $h$ if and only if 
    \begin{equation}\label{condition}
            (e_m^*+e_n^*)  \frac{g'(e_m^* -e_n^* + h)}{g(e_m^* -e_n^* + h)} \leq \frac{c''(e_m^*)/\theta_m - c''(e_n^*)/ \theta_n}{c''(e_m^*)/\theta_m + c''(e_n^*)/\theta_n}.
    \end{equation}
\end{lemma}

\begin{remark}
For quadratic costs $c(e)=e^2/2$, this condition simplifies to
\begin{equation*}\label{condition'}
    (e_m^*+e_n^*)  \frac{g'(e_m^* -e_n^* + h)}{g(e_m^* -e_n^* + h)} \leq \frac{\theta_n-\theta_m}{\theta_n+\theta_m},
\end{equation*}
which holds if $\theta_n>\theta_m$ and the variance is sufficiently large. 
In the uniform--quadratic setting, the condition is equivalent to $\theta_n\geq\theta_m$.

For the normal distribution $G \sim \mathcal{N}(0,\sigma^2)$ and $c(e) = e^2/2$, the condition is equivalent to
$$(e_m^*+e_n^*)(e_n^*-e_m^* - h)  \leq \frac{\theta_n-\theta_m}{\theta_n+\theta_m}\sigma^2,$$
which holds for all $h\geq \bar h(\alpha)$ if $\theta_n>\theta_m$ and $\sigma^2$ is large.
\end{remark}

\begin{remark}\label{rmkl}
    Lemma~\ref{lemma:3} also implies that, even without managerial discretion (i.e., when abilities are exogenous), it can be suboptimal to give a head start to the lower-ability agent in tournaments with output-dependent prizes, in contrast to the common wisdom of ``leveling the playing field.'' 

    Indeed, when $\theta_n=\bar\theta$ is fixed, handicapping the manager instead may generate higher total output.
    For each $\alpha$, the loss in profit from choosing the head start $\bar h(\alpha)$ rather than the output-maximizing head start for fixed agent types captures the opportunity cost of delegating the hiring decision to the manager relative to the full information benchmark.
    \footnote{
    Suppose that $c'''(\cdot)\geq 0$, $g(\cdot)$ is log-concave, and the variance is sufficiently large so that the comparative statics in Lemma~\ref{lemma:1} hold. Then, Lemma~\ref{lemma:3} implies that there exists a cutoff $\tilde h$ (possibly equal to $-\bar H$) such that condition~\eqref{condition} holds if and only if $h\geq \tilde h$. Thus, when abilities are fixed, $\tilde h$ characterizes the output-maximizing head start.
    }
\end{remark}
By Lemma~\ref{lemma:3}, if the regularity condition~\eqref{condition} holds for all $h> \bar h(\alpha)$, the discouragement effect dominates the encouragement effect when abilities are exogenous. Therefore, any further head start beyond $\bar h(\alpha)$ decreases total effort.

\begin{example*}[Uniform--Quadratic]
    If $G \sim \operatorname{Unif}[-\frac{M}{2},\frac{M}{2}]$ and $c(e)=e^2/2$, condition~\eqref{condition} is satisfied if $\theta_n\geq \theta_m$.
    Therefore, once the highest-ability agent $\theta_n=\bar\theta$ is hired by the manager, any further head start given to the manager will decrease total output because the discouragement effect dominates the encouragement effect, that is,
\begin{equation*}
    \frac{d(e_m^*+ e_n^*)}{dh}  = \underbrace{\frac{1}{M/\alpha \theta_m-2}}_{\text{Encouragement effect}} - \underbrace{\frac{1}{M/\alpha\bar \theta-2}}_{\text{Discouragement effect}} \leq 0 \mbox{ when } h > \bar h(\alpha).
\label{eq: EEDE}
\end{equation*}
The inequality is strict if $\theta_m<\bar\theta$.
Thus, the optimal head start is $\bar h(\alpha) = 2\alpha\bar\theta-M/2$.
\end{example*}

Consequently, Lemmas \ref{lemma:2} and \ref{lemma:3} imply the following proposition on the optimal head start.
\begin{proposition} %
    \label{prop:1}
    For any given $\alpha\in(0,1)$, the optimal head start ensures that the manager hires the highest-ability agent. %
    Moreover, if condition~\eqref{condition} holds for all $h> \bar h(\alpha)$ and $\theta_n=\bar\theta$, a head start of $\bar h(\alpha)$ is optimal and is just enough to induce the manager to hire the highest-ability agent.
    \footnote{
        In the uniform--quadratic case, $\bar h(\alpha)$ is the \emph{unique} optimal head start whenever $\theta_m< \bar\theta$.
        More generally, uniqueness holds if condition~\eqref{condition} holds with strict inequality for all $h> \bar h(\alpha)$ and $\theta_n=\bar\theta$.
    }
\end{proposition}

\begin{remark}
    We have assumed that the bound $\bar{H}$ on the head start is sufficiently large so that a head start of $\bar h(\alpha)$ is always feasible.
    Otherwise, if large head starts are infeasible due to fairness concerns, the optimal head start is $\min (\bar h(\alpha), \bar{H})$, and hiring sabotage may still occur.
\end{remark}

In words, the optimal head start always ensures that the manager hires the strongest candidate and thus eliminates hiring sabotage.
Furthermore, under a regularity condition that is satisfied when the variance is sufficiently large, the optimal head start is exactly at the level that induces the manager to hire the strongest candidate.

\section{Optimal Payout Ratio}
In the previous section, we have shown that under some assumptions, for any given $\alpha\in(0,1)$, the optimal head start is $\bar h(\alpha)$, which is just enough to ensure the manager hires the strongest candidate.
The principal also faces a trade-off when determining the optimal payout ratio: a higher payout ratio reduces the principal's share of total output but may incentivize agents to exert more effort.
According to Lemma~\ref{lemma:2}, the payout ratio affects the effort levels of the manager and the new hire in two ways.
On the one hand, a higher payout ratio encourages them to exert more effort. On the other hand, it affects the new hire's ability because the manager may hire a higher or lower-ability agent.

Given the optimal head start $\bar h(\alpha)$ from Proposition~\ref{prop:1}, the optimal payout ratio is given by
\begin{equation}\label{profitmax}
    \alpha^* \in \arg\max_{\alpha\in (0,1)} \Pi(\alpha),\quad \text{where } \Pi(\alpha) =  (1-\alpha ) (e_m^*(\alpha, \bar h(\alpha))+e_n^*(\alpha, \bar h(\alpha))).
\end{equation}
While an interior maximizer \(\alpha^*\in(0,1)\) exists,
its analytical solution and comparative statics are intractable, since the parameters $(\alpha,\sigma, \theta_m)$ affect efforts both directly and indirectly through $\bar h(\alpha)$.
\footnote{
The existence of an interior maximizer follows from the continuity of $\Pi(\alpha)$ on \([0,1]\) (because \(e_m^*(\alpha,h)\), \(e_n^*(\alpha,h)\), and $\bar h(\alpha)$ are all continuous in $\alpha$ and $h$), \(\Pi(0)=\Pi(1)=0\), and \(\Pi(\alpha)>0\) for all \(\alpha\in(0,1)\).
}
For tractability, we henceforth assume a uniform distribution $G \sim \operatorname{Unif}[-\frac{M}{2}, \frac{M}{2}]$
and a quadratic cost function $c(e)=e^2/2$ to remove the strategic interdependence of effort choices in the tournament stage (see also \cite{Konrad2009}, \cite{Ederer2010}, and \cite{BrownMinor2014}).

Given the optimal head start $\bar h(\alpha) = 2\alpha \bar \theta - M/2$, an increase in the payout ratio $\alpha$ has a positive effect on the effort levels of both agents. From the new hire's perspective, since $\theta_n^* = \bar \theta$, an increase in $\alpha$ incentivizes him to exert effort as $e_n^*(\alpha, \bar h (\alpha)) = \alpha \bar \theta$.
From the manager's perspective, an increase in $\alpha$ creates a greater incentive for the manager to both exert effort and engage in sabotage. The higher incentive to sabotage leads to an increase in the optimal head start, which further incentivizes the manager because of the {encouragement effect} of the head start, as $e_m^*(\alpha, \bar h (\alpha)) =  \frac{2\alpha \bar \theta}{M/\alpha\theta_m -2}$.%

\begin{proposition}
    \label{prop:2}
    Suppose that $G\sim\operatorname{Unif}[-\frac{M}{2},\frac{M}{2}]$,
    $c(e)=\frac{e^2}{2}$, and $M > (1 + \sqrt{2})\bar{\theta}$.
    Define 
    \[
    \alpha^* = \left(1+ \sqrt{1-2\frac{\theta_m}{M}}\right)^{-1}, \quad h^* = \bar h(\alpha^*) = 2\alpha^*  \bar \theta - M/2,
    \]
    and suppose $\bar H \geq \left|h^* \right|$.
    We restrict to contracts for which the 
    tournament subgame $\Gamma(\theta_n)$ has a pure-strategy Nash equilibrium for every $\theta_n\in\Theta$.

    Then, the contract $(\alpha^*,h^*)$ is optimal and induces the manager to hire the highest-ability agent in equilibrium.
    \footnote{
        If $\theta_m<\bar\theta$, then $(\alpha^*,h^*)$ is the unique optimal contract.
        }
    The principal's profit is given by
    $\Pi^* = \left( 2\left(\frac{1}{\alpha^*} - \frac{\theta_m}{M}\right)\right)^{-1} \bar\theta$.%
\end{proposition}

\begin{corollary} \label{cor:2.1}
    Under the uniform--quadratic assumption, the following comparative statics hold:
    \begin{enumerate}[label=(\roman*)]
     \item The optimal head start $h^*$ is decreasing in $M$ and increasing in the manager's ability $\theta_m$, and $h^*>0$ if and only if $M<\frac{8\bar\theta^2}{(4\bar\theta-\theta_m)}$.
        
    \item The optimal payout ratio $\alpha^*$ is decreasing in  $M$ and increasing in the manager's ability $\theta_m$. As $M\to \infty$, $\alpha^* \to 0.5$.
    As $M\to 2\bar \theta$, $\alpha^* \to  \left(1+\sqrt{ 1 -\theta_m/ \bar\theta}\right)^{-1}$.
    As $\theta_m\to 0$, $\alpha^*\to 0.5$.

    \item  The optimal profit $\Pi^*$ is decreasing in $M$ and increasing in the manager's ability $\theta_m$. 
    
    \end{enumerate} 
\end{corollary}

As we can see, the optimal payout ratio $\alpha^*$ and head start $h^*$ depend on the manager's ability $\theta_m$ and the noise in the performance evaluation captured by $M$. As the manager's ability $\theta_m$ increases, the manager invests more effort because her marginal cost of effort decreases, so the principal's profit is higher. Meanwhile, when $\theta_m$ is higher, the payout ratio also has a greater marginal effect on the manager's effort, so the optimal payout ratio $\alpha^*$ is higher. The increase in the payout ratio creates a greater incentive to sabotage, thereby requiring the principal to increase the optimal head start $h^*$ to ensure that the new hire is of the highest ability.

In equilibrium, the principal offers a head start to the manager that ensures the highest-ability agent is hired. 
As the noise in performance evaluation increases, the noise begins to take over performance attribution. As a result, the manager becomes less fearful of competing against a higher-ability candidate. 
Thus, the manager has a lower incentive to sabotage, which allows the principal to reduce the optimal head start $h^*$. 

Moreover, for a given \emph{optimal} head start, 
the increase in noise makes the effort less important in determining the winner due to noise in performance measurement, thereby reducing the marginal effect of the payout ratio on the manager's effort.
In other words, for a marginal increase in the payout ratio, the manager's effort increases less than it would if noise were lower. 
Therefore, under higher noise, the principal will lower the optimal payout ratio $\alpha^*$, which in turn decreases the optimal head start even further.
Furthermore, higher noise also reduces the manager's marginal return to effort and thus has a demotivating effect on her equilibrium effort, thereby adversely impacting the principal's profit.

\section{Conclusion}

Hierarchical firms often delegate hiring decisions to managers because managers possess local information about candidates' skills and team-specific fit. This delegation creates a conflict when the selected employee becomes both a coworker and a competitor for performance-based rewards. For example, a senior employee at a technology company may be given authority to select a junior colleague for the team while anticipating future competition for bonuses tied to team performance.

Existing tournament literature does not consider cases where one player can choose the ability of another. We fill this gap by studying the optimal design of a two-player Lazear-Rosen-style tournament, where the manager has discretion over hiring the other player, and the winner receives a fraction of total output. Tournament incentives may lead hiring managers to hire a lower-ability agent than would otherwise be optimal for the firm to avoid future competition.

To mitigate hiring sabotage, the principal designs a head start---an advantage to the manager in the output comparison---and a payout ratio, which is the share of total output awarded to the winner. This role of head starts builds on the literature on biased contests, where head starts (or handicaps) are used to restore efficiency or provide incentives in tournaments.
The output-dependent prize makes the new hire also a coworker and provides some incentives for the manager to hire a high-ability candidate.

We find that the head start has three effects on the output: (i) an encouragement effect on the manager, (ii) a discouragement effect on the new hire, and (iii) a hiring effect through the increased ability of the new hire. The hiring effect dominates the discouragement effect until the strongest candidate is hired; once the strongest candidate is hired, any further head start leads the discouragement effect to dominate the encouragement effect.
Therefore, the optimal contract offers a head start just enough to induce the manager to hire the strongest candidate.
Consequently, we derive the optimal head start and payout ratio that maximize the principal's profit.

In reality, a large head start may be infeasible due to concerns about fairness.
In this case, our model predicts that hiring sabotage may exist in equilibrium because the head start is limited to an insufficient level.
Similarly, head starts may incur additional psychological costs to new employees by lowering their morale, which can also reduce the optimal head start to a level that does not completely eliminate sabotage.
Moreover, the principal may want a higher-ability agent to win the tournament, as she wants to select the strongest agent {through} the tournament and maximize \emph{future} profits (\cite{clark2001rank, hvide2003risk, munster2007selection, ryvkin2008predictive, BrownMinor2014, DrugovRyvkin2017}).
In this case, providing a head start to the manager risks promoting a less competent agent, which harms the firm’s future profitability. Consequently, the optimal head start may allow for some degree of hiring sabotage in equilibrium (i.e., $\theta_n^*(h^*) < \bar\theta$), but it will always ensure that the new hire is more capable than the manager (i.e.,  $\theta_n^*(h^*) >\theta_m$).
This is consistent with \cites{Kawasaki2015} advice to have managers hire employees better than they are.
We formalize this result in Appendix~\ref{twoperiod}.

Our analysis takes full managerial discretion over hiring as a benchmark, highlighting the trade-off between the informational benefits of delegation and the incentive costs of discretion. In practice, senior managers or external evaluators may review hiring decisions. Such oversight or centralized hiring can mitigate hiring sabotage, but it also neglects the manager's local knowledge \citep{DellerSandino2020}. Although we do not formally model partial discretion, the incentive problem we identify remains relevant whenever the incumbent manager retains influence over the hiring decision.

\newpage

\counterwithin*{equation}{section}
\renewcommand{\theequation}{\thesection.\arabic{equation}}

\counterwithin*{theorem}{section}
\renewcommand{\thetheorem}{\thesection.\arabic{theorem}}

\counterwithin*{lemma}{section}
\renewcommand{\thelemma}{\thesection.\arabic{lemma}}

\counterwithin*{proposition}{section}
\renewcommand{\theproposition}{\thesection.\arabic{proposition}}

\counterwithin*{remark}{section}
\renewcommand{\theremark}{\thesection.\arabic{remark}}

\counterwithin*{figure}{section}
\renewcommand{\thefigure}{\thesection.\arabic{figure}}

\begin{appendices}
\onehalfspacing

\counterwithin*{equation}{section}
\renewcommand{\theequation}{\thesection.\arabic{equation}}

\counterwithin*{theorem}{section}
\renewcommand{\thetheorem}{\thesection.\arabic{theorem}}

\counterwithin*{lemma}{section}
\renewcommand{\thelemma}{\thesection.\arabic{lemma}}

\section{Proofs}
\label{app:proofs}

\subsection{Preliminaries}

\begin{lemma} \label{lem:bounded}
    The equilibrium effort levels $e_m$ and $e_n$ are bounded by $\bar e = \sup\{e: c(e)/e\leq 2\bar\theta\}$.
\end{lemma}
\begin{proof}
    First, a unique $\bar e = \sup\{e: c(e)/e\leq 2\bar\theta\} \in (0,\infty)$ exists by the assumptions on $c(e)$ (continuously differentiable, strictly increasing, strictly convex, and $\lim_{e\to\infty} c'(e)>2\bar\theta$).
    
    Now we prove the lemma by contradiction. Suppose WLOG that $e_m>\bar e$.
    Because $m$ can always guarantee a nonnegative payoff by investing $e_m=0$, we have $u_m = \alpha G(e_m-e_n+h) (e_m+e_n) - c(e_m)/\theta_m \geq 0$.
    Because $\alpha\leq1$ and $G(e_m-e_n+h)\leq1$, this implies
    \[\frac{c(e_m)}{e_m}
    \leq \left(1+\frac{e_n}{e_m}\right)\theta_m.\]
    Since $e_m>\bar e$, we have 
    $$\frac{c(e_m)}{e_m} >2\bar\theta \geq 2\theta_m.$$
    Combining this with the previous inequality, we have $e_n>e_m$.
    
    On the other hand, we also have $u_n\geq0$ and thus
    \[
    u_m+u_n = \alpha\cdot (e_m+e_n)-\frac{c(e_m)}{\theta_m}-\frac{c(e_n)}{\theta_n} \geq0.
    \]
    Because  $c(e_m)/\theta_m  > e_m$, we must have  
    $$\frac{c(e_n)}{\theta_n} \leq
(e_m+e_n)-\frac{c(e_m)}{\theta_m}
<
e_n$$  
    and therefore $e_n\leq \bar e <e_m$, contradicting the previous conclusion that $e_n>e_m$. 
\end{proof}

\begin{lemma}
    \label{lem:prelim}
Assume $g(x)$ is continuously differentiable on its support, symmetric, and unimodal at $x=0$. We have the following properties:
\begin{enumerate}[label=(\roman*)]
    \item $g'(x)\leq 0$ for all $x\geq 0$, and $g'(x)\geq 0$ for all $x\leq 0$.
    \item $J(x)\equiv \frac{G(x)}{g(x)}-x \geq \frac{1}{2g(0)}$ and is increasing (decreasing) when $x\geq0$ ($x\leq 0$).
    \item For every $C<\infty$, as the variance of $G$ grows arbitrarily large 
$(\sigma^2\to\infty)$, we have
\begin{align*}
   \sup_{|x|\le C} g(x)\to 0,\quad
\sup_{|x|\le C} |g'(x)|\to 0,\quad
\sup_{|x|\le C}\left|\frac{g'(x)}{g(x)}\right|\to 0,
\\
\sup_{|x|\le C}\left|\frac{g'(x)}{g(x)^2}\right|\to 0,
\quad
\sup_{|x|\le C}|J'(x)|\to 0. 
\end{align*}

\end{enumerate}
\end{lemma}
\begin{proof}
    Part (i) is straightforward because $g(x)$ is single-peaked at zero.
    
    To see (ii), by symmetry, we have $G(0)=1/2$ and $J(0)=\frac{1}{2g(0)}$. Taking the derivative yields
    $$J'(x)=\frac{g(x)^2-G(x)g'(x)}{g(x)^2}-1 =-\frac{G(x)g'(x)}{g(x)^2},$$
    which, by part (i), is nonnegative (nonpositive) for \(x\ge0\) (\(x\le0\)). Thus, $J(x)$ achieves its global minimum at $x=0$.

To see (iii), by our parametric assumption, \(G(x)=F(x/\sigma)\), where \(F\) is a standardized distribution with unit variance and density \(f\), and \(\sigma^2\) is the variance of \(G\). Hence, \(g(x)=f(x/\sigma)/\sigma\) and \(g'(x)=f'(x/\sigma)/\sigma^2\).

Fix $C<\infty$. Since $f$ is continuous and $f(0)>0$, for all sufficiently large $\sigma$,
\[
\inf_{|x|\le C} f(x/\sigma)\geq \frac{f(0)}{2}.
\]
Since $f$ is symmetric and differentiable at zero, $f'(0)=0$.
Continuity of $f'$ at zero therefore implies
\[
\sup_{|x|\le C}|f'(x/\sigma)|\to0
\quad\text{as }\sigma\to\infty.
\]

Therefore, we have the following conclusions as $\sigma \to \infty$.
\[
\sup_{|x|\le C} g(x)
=
\sup_{|x|\le C}\frac{1}{\sigma}f(x/\sigma)
=\mathcal{O}(1/\sigma)\to 0,
\]
\[
\sup_{|x|\le C}
\left|
\frac{g'(x)}{g(x)^2}
\right|
=
\sup_{|x|\le C}
\left|
\frac{f'(x/\sigma)}{f(x/\sigma)^2}
\right|
\leq
\frac{\sup_{|x|\le C}|f'(x/\sigma)|}
{\left(\inf_{|x|\le C} f(x/\sigma)\right)^2}
\leq
\frac{4}{f(0)^2}
\sup_{|x|\le C}|f'(x/\sigma)|
\to 0,
\]
and thus 
\[
\sup_{|x|\le C}
\left|\frac{g'(x)}{g(x)}\right|
\leq    \sup_{|x|\le C}
\left|\frac{g'(x)}{g(x)^2}\right|\cdot \sup_{|x|\le C} g(x)
\to 0,
\]
\[
\sup_{|x|\le C}
\left|g'(x)\right|
\leq    \sup_{|x|\le C}
\left|\frac{g'(x)}{g(x)}\right|\cdot \sup_{|x|\le C} g(x)
\to 0.
\]
Finally,
\[
J'(x)
=
-\frac{G(x)g'(x)}{g(x)^2}.
\]
Since $|G(x)|\leq 1$,
\[
\sup_{|x|\le C}|J'(x)|
\leq
\sup_{|x|\le C}
\left|
\frac{g'(x)}{g(x)^2}
\right|
\to 0.
\]
Therefore, the convergence is uniform over every compact set $|x|\leq C$.
\end{proof}

\begin{lemma}
\label{lem:uniform-curvature}
$c''$ is uniformly
bounded away from zero on $[0,\bar e]$:
there exists ${\kappa}>0$ such that
\[
{\kappa} = \inf_{e\in[0,\bar e]}c''(e).
\]
\end{lemma}
\begin{proof}
By Lemma~\ref{lem:bounded}, $\bar e<\infty$. Since $c$ is twice
continuously differentiable, $c''$ is continuous. It therefore attains
a minimum on the compact interval $[0,\bar e]$. Because $c''(e)>0$ for
every $e\geq0$, this minimum is strictly positive. Hence,
${\kappa} = \inf_{e\in[0,\bar e]}c''(e)>0$.
\end{proof}

\begin{remark} 
\label{rem:weaker-curvature}
The assumption that $c''(e)>0$ for all $e\geq0$ can be weakened to $c''(e)\geq0$ for all $e\geq0$ and $c''(e)>0$ for all $e>0$.  
This allows for cost functions with $c''(0)=0$, such as $c(e)=e^s$ with
$s>2$.

To see this, Lemma~\ref{lem:bounded} and the boundedness of the head
start imply that $|x| = 
|e_m-e_n+h|\leq \bar e+\bar H$.
Since $G(0)=1/2$, we have $G(x)\to\frac12$ uniformly over this bounded interval as $\sigma\to\infty$. Therefore, for sufficiently large $\sigma$, we have $\frac13 \leq G(x)\leq\frac34$ for every $x$ bounded by $\bar e+\bar H$.

Fix $\alpha\in(0,1)$. Because $c'(0)=0$, the marginal payoff from
effort at zero is at least $\alpha/4>0$. 
Hence, equilibrium efforts are interior, and the first-order conditions imply
$c'(e_i) \geq  {\alpha\underline\theta}/{4}$ and therefore
\[
e_i
\geq
\underline e_\alpha
\equiv
(c')^{-1}\left(
\frac{\alpha\underline\theta}{4}
\right)>0,
\qquad i\in\{m,n\}.
\]
Continuity of $c''$ then implies
\[
\kappa_\alpha
\equiv
\inf_{e\in[\underline e_\alpha,\bar e]}c''(e)>0.
\]

Thus, for each fixed $\alpha\in(0,1)$, there exists a variance threshold $\bar\sigma(\alpha)$ such that the relevant large-variance arguments may use $\kappa_\alpha$ in place of $\kappa$.
\end{remark}

\begin{remark}\label{rem:weaker-assumptions}
Instead of assumptions (i) $c''(e)>0$ for all $e\geq0$ and
(ii) the variance of $G$ is sufficiently large, we could instead impose
(i')
\(
\kappa=\inf_{e\in[0,\bar e]}c''(e)
\)
is sufficiently large, which implies that equilibrium efforts satisfy
$e_i=\mathcal O(1/\kappa)$, and (ii') $G$ is log-concave, which implies
\(
J'(x)=(G(x)/g(x))'-1\geq-1.
\)
\end{remark}

\subsection{Proof of Lemma~\ref{lemma:1}}
\begin{proof}
Given $\theta_n$, the first-order conditions are
\begin{align}
   & \frac{\partial u_m}{\partial e_m} \equiv  \alpha[(e_n+e_m)\cdot g(e_m-e_n+h) + G(e_m-e_n+h)] - c'(e_m)/\theta_m = 0, \\
   & \frac{\partial u_n}{\partial e_n} \equiv  \alpha[(e_n+e_m)\cdot g(e_m-e_n+h) + G(e_n-e_m-h)] - c'(e_n)/\theta_n = 0.
\end{align}
When the variance is sufficiently large, $g(x)$ and $|g'(x)/g(x)|$ are sufficiently small, so the second-order conditions are satisfied.

Define $x^*=e_m^*-e_n^*+h$, and
    \begin{align*}
        v_{mm}  &\equiv \frac{\partial^2 u_m}{\partial e_m^2} = \alpha (2g(x^*) + (e_m^*+e_n^*)  g'(x^*)) - c''(e_m^*)/\theta_m   < 0 \\
       v_{nn}  &\equiv \frac{\partial^2 u_n}{\partial e_n^2} = \alpha (2g(x^*) - (e_m^*+e_n^*)  g'(x^*)) - c''(e_n^*)/\theta_n   < 0\\
        v_{mn} &\equiv \frac{\partial^2 u_m}{\partial e_m \partial e_n} = - \alpha   (e_m^*+e_n^*)  g'(x^*)  \\
        v_{nm} &\equiv \frac{\partial^2 u_n}{\partial e_n \partial e_m} =  \alpha  (e_m^*+e_n^*)  g'(x^*)  
     \end{align*}
     Define the Jacobian matrix as
        \[ \mathbf{J}_1 = \begin{pmatrix}
            v_{mm} & v_{mn}\\
            v_{nm} & v_{nn}
        \end{pmatrix}. \]
By Lemma~\ref{lem:bounded}, $e_m$ and $e_n$ are bounded by $\bar e$, so $|x^*|\leq  \bar e + \bar H$.
By Lemma~\ref{lem:prelim}, choose \(\sigma\) sufficiently large that
\[
\sup_{|x|\leq\bar e+\bar H}
\left(2g(x)+2\bar e|g'(x)|\right)
<
 {\kappa}/{\bar\theta}.
\]
Because \(c''(e_i)\ge\kappa\) for \(e_i\in[0,\bar e]\), we have \(v_{mm}<0\) and \(v_{nn}<0\) uniformly. 
Therefore, the determinant of the Jacobian matrix is
\[
\det(\mathbf J_1)= v_{mm}v_{nn}- v_{mn}v_{nm} >0.
\footnote{This implies the Nash equilibrium in the subgame $\Gamma(\theta_n)$ is stable.}
\]

Denote     $  v_{mp} = \frac{\partial^2 u_m}{\partial e_m \partial p}$ and
        $v_{np} = \frac{\partial^2 u_n}{\partial e_n \partial p}$ for any parameter $p\in\{\theta_m, \theta_n, \alpha, h\}$.
By the implicit function theorem,  
\begin{equation}
    \begin{pmatrix}
           d e_m^*/d p\\
           d e_n^*/d p
        \end{pmatrix}
        = - \mathbf{J}_1^{-1} \begin{pmatrix}
           v_{mp}\\
           v_{np}
        \end{pmatrix} 
\end{equation}

Therefore, 
(1) with respect to $\theta_m$, 
because $v_{m,\theta_m}=c'(e_m^*)/\theta_m^2$ and $v_{n,\theta_n}=0$,
we have
\begin{align}
    & \frac{d e_m^*}{d\theta_m} %
    = -\frac{c'(e_m^*)}{\theta_m^{2}} \frac{v_{nn}}{ \det(\mathbf{J}_1)}  >0,\\
    & \frac{d e_n^*}{d\theta_m} %
     = \frac{c'(e_m^*)}{\theta_m^2} \frac{ \alpha (e_m^*+e_n^*)  g'(x^*) }{ \det(\mathbf{J}_1)}  \stackrel{\text{sign}}{=} g'(x^*).
 \end{align}
 By Lemma~\ref{lem:prelim}, $g'(x)\leq0$ ($g'(x)\geq0$) if $x\geq 0$ ($x\leq 0$) because $g(x)$ is unimodal. Thus, $\frac{d e_n^*}{d\theta_m}\geq0$ ($\frac{d e_n^*}{d\theta_m}\leq0$) if $n$ is more (less) likely to win.
 By symmetry, with respect to $\theta_n$, we also have
 \begin{align}
    & \frac{d e_n^*}{d\theta_n} = -\frac{c'(e_n^*)}{\theta_n^{2}} \frac{v_{mm}}{ \det(\mathbf{J}_1)} >0,\\
    & \frac{d e_m^*}{d\theta_n}  = -\frac{c'(e_n^*)}{\theta_n^2} \frac{ \alpha (e_m^*+e_n^*)  g'(x^*) }{ \det(\mathbf{J}_1)} \stackrel{\text{sign}}{=} -g'(x^*).
 \end{align}
As shown above, $\frac{d e_m^*}{d\theta_n}\geq0$ ($\frac{d e_m^*}{d\theta_n}\leq0$) if $m$ is more (less) likely to win.

(2) With respect to $\alpha$, because 
 $v_{m \alpha}=(e_m^*+e_n^*) g(x^*)+G(x^*)=\frac{c'(e_m^*)}{\alpha \theta_m}$ and  $v_{n \alpha}=(e_m^*+e_n^*) g(x^*)+G(-x^*)=\frac{c'(e_n^*)}{\alpha \theta_n}$,
we have
 \begin{align}
    & \frac{d e_m^*}{d\alpha} %
    =   \frac{ -v_{nn}c'(e_m^*)/\alpha\theta_m - (e_m^*+e_n^*)  g'(x^*) c'(e_n^*)/\theta_n }{ \det(\mathbf{J}_1)}   >0,\\
    & \frac{d e_n^*}{d\alpha} %
     =   \frac{  -v_{mm}c'(e_n^*)/\alpha\theta_n + (e_m^*+e_n^*)  g'(x^*) c'(e_m^*)/\theta_m }{ \det(\mathbf{J}_1)}>0,
 \end{align}

(3) With respect to $h$, because
$v_{m h} = \alpha  ( (e_m^* + e_n^*) g'(x^*) + g(x^*)  )$ and $v_{n h} = \alpha  ( (e_m^* + e_n^*) g'(x^*) - g(x^*)  )$,
we have
\begin{align}
   & \frac{d e_m^*}{dh} %
   =   \alpha g(x^*) \frac{-2\alpha g(x^*)  +  \left( 1+ (e_m^*+e_n^*)  \frac{g'(x^*)}{g(x^*)} \right) c''(e_n^*)/ \theta_n  }{  \det(\mathbf{J}_1)}   >0,\\
   & \frac{d e_n^*}{dh} %
    =   \alpha g(x^*) \frac{  2\alpha g(x^*)  - \left( 1- (e_m^*+e_n^*)  \frac{g'(x^*)}{g(x^*)} \right) c''(e_m^*)/\theta_m}{  \det(\mathbf{J}_1)}<0,
\end{align}
when the variance of $G$ is sufficiently large.%

(4) With respect to $\sigma$ within a scale family $G(x) = F(x/\sigma)$, 
because 
\begin{align*}
& v_{m \sigma} = -\frac{\alpha}{\sigma} ((e_m^*+e_n^*) (g(x^*)+x^* g^{\prime}(x^*) )+x^* g(x^*) ), \\ 
& v_{n \sigma} = -\frac{\alpha}{\sigma} ((e_m^*+e_n^*) (g(x^*)+x^* g^{\prime}(x^*) )-x^* g(x^*) ),
\end{align*}
we have
 \begin{align}
    & \frac{d e_m^*}{d\sigma} = - \frac{\alpha g(x^*)^2}{\sigma} \frac{ \left( (2e_m^*+h) + (e_m^*+e_n^*)x^*\frac{g'(x^*)}{g(x^*)} \right) c''(e_n^*)/\theta_n - 2\alpha g(x^*) (2e_m^*+h) }{ \det(\mathbf{J}_1)},  \\
    & \frac{d e_n^*}{d\sigma} =  - \frac{\alpha g(x^*)^2}{\sigma} \frac{ \left( (2e_n^*-h) + (e_m^*+e_n^*)x^*\frac{g'(x^*)}{g(x^*)} \right) c''(e_m^*)/\theta_m - 2\alpha g(x^*) (2e_n^*-h) }{ \det(\mathbf{J}_1)},
 \end{align}
 For sufficiently large variance of $G$,
 $\frac{d e_m^*}{d\sigma}<0$ if $2e_m^*+h>0$, and $\frac{d e_n^*}{d\sigma}<0$ if $2e_n^*-h>0$.
 Sufficient conditions are $h\geq0$ and $h\leq 0$, respectively.

(5)
We show that if $\theta_m\geq\theta_n$ ($\theta_m\leq\theta_n$) and $h\geq0$ ($h\leq0$), then $x^* \geq 0$ ($x^*\leq 0$), i.e., the manager is more (less) likely to win.

First, note that when $h=0$ and $\theta_m=\theta_n$, we have $e_m^*=e_n^*$ and thus $x^*=0$.

Then, because
\begin{align}
    \frac{d x^*}{d\theta_m }  =  \frac{d e_m^*}{d\theta_m } -  \frac{d e_n^*}{d\theta_m }
    =  \frac{c'(e_m^*)}{\theta_m^{2}} \frac{c''(e_n^*)/\theta_n - 2\alpha g(x^*)}{ \det(\mathbf{J}_1)} >0,
 \end{align}
$\frac{d  x^*}{d\theta_n } <0 $ (by symmetry), and
    \begin{align}
        \frac{d  x^*}{d h}  
        =  \frac{d e_m^*}{dh} - \frac{d e_n^*}{dh} + 1 > 0,
    \end{align}
we have $x^*\geq0$ if $\theta_m\geq\theta_n$ and $h\geq0$.
By symmetry, $x^*\leq0$ if $\theta_n\geq\theta_m$ and $h\leq0$.

(6) Finally, because efforts are bounded, there exist \(\bar h_1,\bar h_2>0\) such that \(x^*\ge0\) for all \(h>\bar h_1\) and \(x^*\le0\) for all \(h<-\bar h_2\).
\end{proof}

\subsection{Proof of Lemma~\ref{lemma:2}}

\begin{proof}

    Denote $x(\theta_n)=e_m(\theta_n)-e_n(\theta_n)+h$ and $J(x)= G(x)/g(x)-x$.
    Then, for all $\theta_n\in\Theta$, we have
    \begin{align}  \label{foc:theta_n}
        \frac{d  u_m (e_m(\theta_n),e_n(\theta_n),\theta_n)}{d \theta_n} =  
        \alpha e_n'(\theta_n)  [J(x(\theta_n)) - (2e_n(\theta_n)-h)] g(x(\theta_n)).
    \end{align}
    First, we show that $u_m$ is quasiconcave in $\theta_n$.
    Define $z(\theta_n)=J(x(\theta_n)) - (2e_n(\theta_n)-h)$. We have
    \begin{equation*}
        \begin{aligned}
        z'(\theta_n) 
        &= J'(x(\theta_n)) x'(\theta_n) - 2 e_n'(\theta_n)\\
        &= \frac{c'(e_n)}{\theta_n^2 \det(\mathbf{J}_1)}\left(2 \alpha\left[\left(e_m+e_n\right) g^{\prime}(x)+\left(J^{\prime}(x)+2\right) g(x)\right]-\left(J^{\prime}(x)+2\right) \frac{c^{\prime \prime}\left(e_m\right)}{\theta_m}\right)
        \end{aligned}
    \end{equation*}
    where \(e_m=e_m(\theta_n)\), \(e_n=e_n(\theta_n)\), and 
\(x=x(\theta_n)\).

    To determine the sign of $z'(\theta_n)$ when $\sigma$ is sufficiently large, first note that, by Lemma~\ref{lem:bounded}, $e_m(\theta_n)$ and $e_n(\theta_n)$ are bounded by $\bar e$ for all $\theta_n\in\Theta$. Because $|h|\leq\bar H$, we have $|x(\theta_n)|=|e_m(\theta_n)-e_n(\theta_n)+h|\leq  \bar e + \bar H$ for all $\theta_n\in\Theta$.  
    Therefore, by Lemma~\ref{lem:prelim}, we have
    \[
    \sup_{\theta_n\in\Theta}|g(x(\theta_n))|\to0,
    \quad
    \sup_{\theta_n\in\Theta}|g'(x(\theta_n))|\to0,
    \quad
    \sup_{\theta_n\in\Theta}|J'(x(\theta_n))|\to0
    \]
    as $\sigma \to \infty$.
Now we decompose the bracketed term into two parts:
\[
\begin{aligned}
    B_1(\theta_n)
&= 2\alpha\left[
(e_m(\theta_n)+e_n(\theta_n))g'(x(\theta_n))+(J'(x(\theta_n))+2)g(x(\theta_n))
\right]\\
B_2(\theta_n)
&=
-(J'(x(\theta_n))+2)\frac{c''(e_m(\theta_n))}{\theta_m}.
\end{aligned}
\]
As $\sigma \to \infty$, we have 
$$
\sup_{\theta_n\in\Theta}|B_1(\theta_n)| \to 0,\quad \mbox{and }
\sup_{\theta_n\in\Theta}\left|B_2(\theta_n)+2\frac{c''(e_m(\theta_n))}{\theta_m}\right| \to 0.$$
    Because $c''(e_m(\theta_n))\geq {\kappa}>0$ for all $\theta_n\in\Theta$, there exists \(\bar\sigma<\infty\) such that, for all 
    \(\sigma\ge \bar\sigma\),
    \[
    B_1(\theta_n) + B_2(\theta_n) <0
    \quad
    \text{for all }\theta_n\in\Theta.
    \]
    Thus, we have $z'(\theta_n) <0$ for all $\theta_n\in\Theta$ when the variance is sufficiently large.

    Hence, $u_m$ is quasiconcave in $\theta_n$, and there exists a unique $\theta_n^*\in \Theta$ (possibly at the boundary) that satisfies $u_m'(\theta_n)\geq0$ if and only if $\theta_n\leq \theta_n^*$.
    In this lemma, we focus on the interior solution $\theta_n^*\in(\underline\theta,\bar\theta)$ such that $z(\theta_n^*)=0$, $u_m'(\theta_n^*)=0$, and 
    $2e_n(\theta_n^*) -h = J(x(\theta_n^*))>0$.

    When $\theta_n^*\in(\underline\theta,\bar \theta)$, the first-order conditions are
    \begin{align}
        &c'(e_n^* )/\theta_n^*=\alpha, \\
        &c'(e_m^* )/\theta_m  = 2\alpha  \cdot G(e_m^* -e_n^* +h), \\
        &G(e_m^*-e_n^*+h) = (e_m^*+e_n^*)\cdot g(e_m^*-e_n^*+h).
    \end{align}
    Denote $e_m^*=e_m(\theta_n^*)$, $e_n^*=e_n(\theta_n^*)$, and $x^*=x(\theta_n^*)$.
    Denote $w_m = c'(e_m^* )/\theta_m - 2\alpha  \cdot G(e_m^* -e_n^* +h)$
    and $w_n= G(e_m^*-e_n^*+h)-(e_m^*+e_n^*)\cdot g(e_m^*-e_n^*+h)$.
    Define
    $w_{m\cdot}$ and $w_{n\cdot}$ analogously to $v_{m\cdot}$ and $v_{n\cdot}$ in the proof of Lemma~\ref{lemma:1}.
    
    Denote the Jacobian matrix as
        \[ \mathbf{J}_2  
        =   \begin{pmatrix}
            w_{mm} & w_{mn}\\
            w_{nm} & w_{nn}
        \end{pmatrix}
        =\begin{pmatrix}
            c''(e_m^*)/\theta_m - 2\alpha g(x^*) & 2\alpha g(x^*)\\
            -(e_m^*+e_n^*)g'(x^*) & -2g(x^*)+(e_m^*+e_n^*)g'(x^*)
        \end{pmatrix}. \]
    By the implicit function theorem, for any parameter $p\in\{\theta_m, \alpha, h\}$,
    \begin{equation}
        \begin{pmatrix}
               d e_m^*/d p\\
               d e_n^*/d p 
            \end{pmatrix}
            = - \mathbf{J}_2^{-1}  
            \begin{pmatrix}
               w_{mp}\\
               w_{np} 
            \end{pmatrix}
    \end{equation}
    and 
    \[
    \frac{d\theta_n^*}{d p} = \frac{c''(e_n^*)}{\alpha} \frac{d e_n^*}{d p} -  \frac{c' (e_n^*)}{\alpha^2}  \frac{d \alpha}{d p}.
    \]

    Similar to Lemma~\ref{lemma:1}, when the variance is sufficiently large, we have 
    \begin{align}
    & A \equiv  {w_{mm}}  = \frac{c''(e_m^*)}{ \theta_m} -2\alpha g(x^*)>0,\\
    & B \equiv   -\frac{w_{nn}}{g(x^*)} =  2 - (e_m^*+e_n^*) \frac{g'(x^*)}{g(x^*)} >0,\\
    & {D} \equiv -\frac{\det(\mathbf{J}_2)}{ g(x^*)} =  \frac{c''(e_m^*)}{ \theta_m} B -4\alpha g (x^*) >0 ,\\
    & {D}-A = \frac{c''(e_m^*)}{ \theta_m} (B-1) -2\alpha g (x^*)  >0
    \end{align}
    and thus second-order conditions are satisfied.

   (1) With respect to $\theta_m$, 
    because $w_{m,\theta_m}=-c'(e_m^*)/\theta_m^2$ and $w_{n,\theta_m}=0$,
    we have
    \begin{align}
        & \frac{d e_m^*}{d\theta_m} =  \frac{ B }{{D}} \frac{c'(e_m)}{\theta_m^2}  >0, \\
        & \frac{d e_n^*}{d\theta_m} = - \frac{1 }{{D}}(e_m^*+e_n^*) \frac{c'(e_m)}{\theta_m^2}  \frac{g'(x^*)}{g(x^*)} \stackrel{\text{sign}}{=} -g'(x^*), \\
        & \frac{d \theta_n^*}{d\theta_m} =   \frac{c''(e_n^*)}{\alpha}  \frac{d e_n^*}{d\theta_m} \stackrel{\text{sign}}{=} -g'(x^*).
    \end{align}
    By Lemma~\ref{lem:prelim}, $g'(x)\leq0$ ($g'(x)\geq0$) if $x\geq 0$ ($x\leq 0$) because $g(x)$ is single-peaked at zero. Thus, $\frac{d e_n^*}{d\theta_m}$ and $\frac{d \theta_n^*}{d\theta_m}$ are positive (negative) if $m$ is more (less) likely to win.

   (2) With respect to $\alpha$, 
   because $w_{m \alpha} = -2  G(x^*)= - c'(e_m^*)/\alpha\theta_m$ and $w_{n \alpha} =0$,
   we have
    \begin{align}
        & \frac{d e_m^*}{d\alpha} =  \frac{ B }{{D}} \frac{c'(e_m)}{\alpha \theta_m}  >0, \\
        & \frac{d e_n^*}{d\alpha} = - \frac{1 }{{D}}\frac{c'(e_m)}{\alpha \theta_m}  (e_m^*+e_n^*) \frac{g'(x^*)}{g(x^*)} \stackrel{\text{sign}}{=} -g'(x^*), \\
        & \frac{d \theta_n^*}{d\alpha} =   \frac{c''(e_n^*)}{\alpha}  \frac{d e_n^*}{d\alpha} - \frac{c'(e_n^*)}{\alpha ^2}<0.
    \end{align}
    The last equation is negative if $g'(x^*)\geq 0$, because then $d e_n^*/d\alpha\leq 0$.
    When $g'(x^*)\leq 0$, it holds under the large-variance condition, since $|g'(x^*)/g(x^*)|$ is sufficiently small and hence $d e_n^*/d\alpha$ is sufficiently small.

    (3)  With respect to $h$, 
    because $w_{m h} = -2\alpha g(x^*)$ and $w_{n h} = g(x^*) - (e_m^*+e_n^*)g'(x^*)$,
    we have
    \begin{align}
        & \frac{d e_m^*}{dh} =   \frac{2\alpha g(x^*)}{{D}}>0, \\
        & \frac{d e_n^*}{dh} =   \frac{{D}-A}{{D}}>0, \\
        & \frac{d \theta_n^*}{dh} =  \frac{c''(e_n^*)}{\alpha}  \frac{d e_n^*}{dh}>0.
    \end{align}
    Moreover,  holding the new hire's ability fixed, we have %
   \begin{equation}
    \frac{\partial e_n(\theta_n^* ,h)}{\partial h} 
    = - \alpha g(x^*) \frac{ {D}-A}{\det(\mathbf{J}_1)} <0.
   \end{equation}
    Thus, the discouragement effect is negative but is dominated by the hiring effect.

   (4) With respect to $\sigma$ and within a scale family $G(x) = F(x/\sigma)$,
because $w_{m\sigma} = 2\alpha x^* g(x^*)/\sigma$
and $w_{n\sigma} =  ((e_m^*+e_n^*) (g(x^*)+x^* g^{\prime}(x^*) )-x^* g(x^*) )/\sigma$,
we have
  \begin{align}
    & \frac{d e_m^*}{d\sigma} =   - \frac{2\alpha(2e_m^*+h)g(x^*)}{\sigma {D}}<0 \quad \mbox{if }h\geq0, \\
    & \frac{d e_n^*}{d\sigma} =    \frac{A \cdot J(x^*) +  \frac{c''(e_m^*)}{ \theta_m}  (e_m^*+e_n^*) x^*\frac{g'(x^*)}{g(x^*)}}{ \sigma {D}}> 0, \\
    & \frac{d \theta_n^*}{d\sigma} = \frac{c''(e_n^*)}{\alpha}  \frac{d e_n^*}{d\sigma}>0
\end{align}
The second inequality follows from $A>0$ and
$J(x^*)\geq \frac1{2g(0)}>0$. The remaining term is dominated under the
large-variance condition because $|g'(x^*)/g(x^*)|$ is sufficiently small.

(5)
We show that
there exist $\bar\theta_m, \underline\theta_m\in \Theta$ such that if $\theta_m>\bar\theta_m$ ($\theta_m<\underline\theta_m$) and $h\geq0$ ($h\leq0$),
then $x^* \geq 0$ ($x^*\leq0$), i.e., the manager is more (less) likely to win.

First, note that when $h=0$ and $\theta_m\to0$, we have $e_m^*\to0$, so the first-order condition~\eqref{foc:theta_n} implies
    \begin{equation}
       G(-e_n^*) = e_n^* g(-e_n^*) \iff  e_n^*=\frac{1-G(e_n^*)}{g(e_n^*)} >0.  
    \end{equation}
    Therefore, $x^*=-e_n^*<0$.
    On the other hand, when $h=0$ and $\theta_m=\bar\theta$, we have $\theta_n^*\leq \theta_m$.
    Therefore, by Lemma~\ref{lemma:1}, $x^*\geq0$.

    Then, because 
    \begin{align}
        \frac{d x^*}{d\theta_m }  =  \frac{d e_m^*}{d\theta_m } -  \frac{d e_n^*}{d\theta_m }
        =  \frac{c'(e_m^*)}{\theta_m^{2}} \frac{2}{ D} >0,
     \end{align}
    and
        \begin{align}
            \frac{d  x^*}{d h}  
            =  \frac{d e_m^*}{dh} - \frac{d e_n^*}{dh} + 1
            = \frac{ 2\alpha g(x^*) +A }{{D}}>0,
        \end{align}
    we have $x^*\geq0$ if $\theta_m\geq\bar\theta_m$ and $h\geq0$.
    Similarly, we also have $x^*\leq0$ if $\theta_m\leq\underline\theta_m$ and $h\leq0$.

(6) Finally, because equilibrium efforts are bounded,
there exist $\bar h_1,\bar h_2>0$ such that $x^* = e_m^*-e_n^*+h \geq0$ for all $h>\bar h_1$ and $x^*= e_m^*-e_n^*+h  \leq0$ for all $h<- \bar h_2$.
\end{proof}

\subsection{Proof of Lemma~\ref{lemma:3}}
\begin{proof}
    Given $\theta_n$, the first-order conditions are
    \begin{align}
       & u_m \equiv  \alpha[(e_n+e_m)\cdot g(e_m-e_n+h) + G(e_m-e_n+h)] - c'(e_m)/\theta_m = 0, \\
       & v_n \equiv  \alpha[(e_n+e_m)\cdot g(e_m-e_n+h) + G(e_n-e_m-h)] - c'(e_n)/\theta_n = 0.
    \end{align}
        We have
    \begin{equation}
        \frac{d e_m^*}{dh}+ \frac{d e_n^*}{dh}  = \alpha \frac{ (g(x^*)+(e_m^*+e_n^*)  g'(x^*)) \frac{c''(e_n^*)}{\theta_n}  -(g(x^*)-(e_m^*+e_n^*)  g'(x^*)) \frac{c''(e_m^*)}{\theta_m}}{ \det(\mathbf{J}_1)}. 
    \end{equation}
    Thus,  $\frac{d e_m^*}{dh}+ \frac{d e_n^*}{dh} \leq 0$ if and only if 
    \begin{equation}
        (e_m^*+e_n^*)  \frac{g'(x^*)}{g(x^*)} \leq \frac{c''(e_m^*)/\theta_m - c''(e_n^*)/\theta_n}{c''(e_m^*)/\theta_m + c''(e_n^*)/\theta_n}.
    \end{equation}
    With quadratic costs $c(e)=e^2/2$, this condition simplifies to
    \begin{equation} 
        (e_m^*+e_n^*)  \frac{g'(e_m^* -e_n^* +h)}{g(e_m^* -e_n^* +h)} \leq \frac{\theta_n-\theta_m}{\theta_n+\theta_m},
    \end{equation}
    which is satisfied if $\theta_n>\theta_m$ and the variance is sufficiently large because $|g'(x)/g(x)|$ is small.
    \end{proof}

\subsection{Proofs of Proposition~\ref{prop:2} and Corollary}
\begin{proof}[Proof of Proposition~\ref{prop:2}]
By Proposition~\ref{prop:B1}, which covers winner-take-all tournaments as a special case, the principal's profit satisfies
\begin{equation}  
   \Pi(\alpha,h)
\leq
(1-\alpha)
\frac{\alpha\bar\theta M}
     {M-2\alpha\theta_m}
\equiv\overline\Pi(\alpha).  
\end{equation}
   Maximizing $\overline\Pi(\alpha) =  
    (1-\alpha ) (\alpha \bar \theta +  \frac{2\alpha\bar \theta}{M/\alpha\theta_m -2})$ gives \[\alpha^* = \frac{M - \sqrt{M(M-2\theta_m)}}{2\theta_m} = \left(1+ \sqrt{1-2\frac{\theta_m}{M}}\right)^{-1}.\]
   The head start \(h^* = \bar h (\alpha^*) = 2 \alpha^* \bar\theta - M/2\), which induces the manager to choose $\theta_n^* = \bar\theta$, is feasible by the assumption $\bar H\geq|h^*|$.
   The contract $(\alpha^*, h^*)$ attains the upper bound $\overline\Pi(\alpha^*)$.
   By Proposition~\ref{prop:1}, conditional on
    $\alpha^*$, $h^*=\bar h(\alpha^*)$ is also the unique optimal head start when
    $\theta_m<\bar\theta$.

    The principal's profit is 
\[
\Pi^*=\overline\Pi(\alpha^*)
= \frac{M(M-\theta_m-\sqrt{M(M-2\theta_m)})}
        {2\theta_m^2}\bar\theta
= \left(
2\left(\frac1{\alpha^*}-\frac{\theta_m}{M}\right)
\right)^{-1}\bar\theta.
\]
\end{proof}

\begin{proof}[Proof of Corollary \ref{cor:2.1}]
    For $\alpha^* = \frac{M - \sqrt{M(M-2\theta_m)}}{2\theta_m}$, because $M>2\bar\theta\geq 2 \theta_m$, we have
        \begin{align}
            &\frac{d \alpha^*}{d M} = \frac{\sqrt{M(M-2\theta_m)}-(M-\theta_m) }{2\theta_m\sqrt{M(M-2\theta_m)}}<0.\\
           & \frac{d \alpha^*}{d \theta_m} =  - \frac{\sqrt{M(M-2\theta_m)}-(M-\theta_m) }{2\theta_m^2 \sqrt{M(M-2\theta_m)}}M>0.
        \end{align}
    
    For $h^* = 2\alpha^* \bar\theta -M /2 $,
    \begin{align}
        &\frac{d h^*}{d M} = 2\frac{d \alpha^*}{d M} \bar \theta - \frac{1}{2}<0.\\
       & \frac{d h^*}{d \theta_m} = 2\frac{d \alpha^*}{d \theta_m}  \bar \theta>0.
    \end{align}
    Direct calculations show that \(h^*>0\iff M<8\bar\theta^2/(4\bar\theta-\theta_m)\).
    
    Because $\Pi^* = \max_{\alpha} 
    (1-\alpha ) (\alpha \bar \theta +  \frac{2\alpha\bar \theta}{M/\alpha\theta_m -2})$, by the envelope theorem, we have $\frac{d \Pi^*}{d M}<0$ and $\frac{d \Pi^*}{d \theta_m}>0$.
\end{proof}

\section{General Prize Sharing Rules}
\label{app:beta}
\subsection{Setup and Results}

We consider a general sharing rule where the winner receives share $\alpha$ and the loser receives share $\beta$, where $0\leq \beta \leq \alpha$ and $0<\alpha + \beta < 1$, instead of a winner-take-all tournament.
Let $\gamma = \alpha - \beta\in [0,\alpha]$ denote the spread and $\tau = \alpha + \beta \in (0,1)$ denote the total share.
In particular, the winner-take-all tournament corresponds to $\beta=0$ and hence
$\gamma=\tau=\alpha$.
The principal's profit is given by $\Pi(\alpha,\beta) = (1-\tau) (e_m^*+e_n^*)$.

Given the new hire's ability $\theta_n$, in the tournament subgame $\Gamma(\theta_n)$, the expected payoffs for $m$ and $n$ are given by
\begin{align}
   u_m(e_m,e_n,\theta_m) & = \left[\beta + \gamma  G(e_m-e_n+h)\right] (e_m+e_n) - \frac{c(e_m)}{\theta_m}, \\
   u_n(e_m,e_n,\theta_n) & = \left[\beta + \gamma  G(e_n-e_m-h)\right] (e_m+e_n) - \frac{c(e_n)}{\theta_n}.
\end{align}
The first-order conditions for the tournament stage are
\begin{align}
  \frac{\partial u_m}{\partial e_m} &= \gamma\cdot(e_m+e_n) g(e_m-e_n+h) + \beta + \gamma G(e_m-e_n+h) - \frac{c'(e_m)}{\theta_m} = 0, \label{FOC1'}\\
  \frac{\partial u_n}{\partial e_n} &= \gamma\cdot(e_m+e_n) g(e_m-e_n+h) + \beta + \gamma G(e_n-e_m-h) - \frac{c'(e_n)}{\theta_n} = 0. \label{FOC2'}
\end{align}
The equilibrium effort levels $e_m(\theta_n)$ and $e_n(\theta_n)$ in the tournament stage are implicitly defined by the above first-order conditions.

\begin{example*}[Uniform--Quadratic]
Suppose $G \sim \operatorname{Unif}[-\frac{M}{2},\frac{M}{2}]$, $c(e)=e^2/2$, and $M$ is sufficiently large that all winning probabilities are interior.
Then, 
\begin{align*}
    e_m(\theta_n) &= \frac{\tau M/2 + \gamma h}{M/\theta_m -2 \gamma}, \\
    e_n(\theta_n) &=  \frac{\tau M/2 - \gamma h}{M/\theta_n -2 \gamma}.
\end{align*}

Fixing the spread $\gamma$, a larger total share $\tau$ increases both the manager's and the new hire's efforts.
Fixing the total share $\tau$, a wider spread $\gamma$ increases the manager's effort if and only if $h\geq -\tau \theta_m$ and increases the new hire's effort if and only if $h\leq  \tau \theta_n$. %
\end{example*}

Then, at the hiring stage, the manager chooses $\theta_n$ to maximize her expected payoff in the tournament stage, and the first-order condition is 
\begin{align} \label{FOC3'}
   \frac{d u_m}{d \theta_n} =
   \frac{\partial u_m}{\partial e_n} e_n'(\theta_n) 
   = \left[\beta + \gamma G(e_m-e_n+h) - \gamma \cdot (e_m+e_n) g(e_m-e_n+h)\right] e_n'(\theta_n).
\end{align}
Hence, in the SPNE, when the hiring decision $\theta_n^*\in(\underline{\theta},\bar\theta)$, the first-order conditions for $\theta_n^*$ and the equilibrium effort levels $(e_m^*,e_n^*)$ are given by 
\begin{align}
\beta + \gamma G(e_m^*-e_n^*+h) &= \gamma\cdot(e_m^*+e_n^*) g(e_m^*-e_n^*+h), \label{eqm1'} \\
 c'(e_n^*)/\theta_n^* &= \tau, \label{eqm2'} \\
 c'(e_m^*)/\theta_m  &= 2\left[\beta + \gamma G(e_m^*-e_n^*+h)\right]. \label{eqm3'}
\end{align}

If the prize spread $\gamma>0$, the comparative statics of the benchmark model continue to apply if both winning probabilities are interior. When $\theta_n^*\in(\underline\theta,\bar\theta)$, the head start has a positive encouragement effect on the manager (i.e., $de_m^*/dh>0$), and its hiring effect dominates the discouragement effect (i.e., $de_n^*/dh>0$).
Once the strongest candidate $\theta_n^*=\bar\theta$ is hired, total effort is weakly decreasing in $h$ whenever condition~\eqref{condition} holds.

If the prize spread $\gamma=0$ (i.e., $\alpha=\beta$), the manager has no incentive to sabotage hiring, so she always hires the strongest candidate $\theta_n^*=\bar\theta$, and the equilibrium effort levels are given by $c'(e_m^*)=\alpha \theta_m$ and $c'(e_n^*)=\alpha \bar\theta$ regardless of the head start.

\begin{example*}[Uniform--Quadratic]
Suppose that $G \sim \operatorname{Unif}[-\frac{M}{2},\frac{M}{2}]$, $c(e)=e^2/2$, and $M$ is sufficiently large that all winning probabilities are interior. 
If $\gamma>0$, the equilibrium efforts and the optimal hiring decision are given by 
\begin{align*}
    \theta_n^*
    &= \min\left\{
    \frac{\tau M/2+\gamma h}{2\gamma\tau},
    \bar \theta
    \right\}, \\
    e_m^*
    &= \frac{\tau M/2+\gamma h}{M/\theta_m - 2\gamma}, \\
    e_n^*
    &= \min\left\{
    \frac{\tau M/2+\gamma h}{2\gamma},
    \frac{\tau M/2-\gamma h}{M/\bar \theta -2\gamma}
    \right\}.
\end{align*}
If $\gamma=0$ (i.e., $\alpha=\beta=\tau/2$), we have 
$\theta_n^*=\bar\theta$, $e_m^*=\alpha \theta_m $, and $e_n^*=\alpha \bar\theta$.

Fixing the spread $\gamma$, a higher total share $\tau$ always increases the manager's effort.
It also increases her incentive to sabotage if and only if $h\geq0$, and the aggregate effect on the new hire's effort is always positive until the strongest candidate is hired (i.e., if $\theta_n^* < \bar\theta$). 

Fixing the total share $\tau$, a wider spread $\gamma$ increases the manager's effort if and only if $h\geq -\tau \theta_m$. It also increases her incentive to sabotage, and the aggregate effect on the new hire's effort is negative until the strongest candidate is hired.
\end{example*}

Although awarding a positive share $\beta>0$ to the loser can
mitigate hiring sabotage, this may be unnecessary when the principal
can instead rely on a head start to serve the same purpose. Indeed, the
following proposition shows that, in the uniform--quadratic setting,
the optimal sharing rule is winner-take-all---that is, $\beta^*=0$.

\begin{proposition} \label{prop:B1}
 Suppose
$G\sim\operatorname{Unif}[-\frac{M}{2},\frac{M}{2}]$, $c(e)=\frac{e^2}{2}$, and $
M>(1+\sqrt{2})\bar\theta.$
Define
\[
\alpha^*
=
\left(
1+\sqrt{1-\frac{2\theta_m}{M}}
\right)^{-1}, \quad h^* = 2\alpha^*\bar\theta-\frac{M}{2},
\]
and suppose $\bar H \geq \left|h^* \right|$.
We restrict to contracts for which the 
tournament subgame $\Gamma(\theta_n)$ has a pure-strategy Nash equilibrium for every $\theta_n\in\Theta$.

Then, the optimal prize-sharing rule is winner-take-all (\(\beta^*=0\)), the optimal payout ratio and head start are \((\alpha^*,h^*)\), and this contract induces the manager to hire the highest-ability agent.
\footnote{
If $\theta_m<\bar\theta$, then $(\alpha^*,\beta^*,h^*)$ is the unique optimal contract.
}
\end{proposition}

The proof is in Appendix~\ref{app:proofs-beta}.
Our results show that giving a head start $h$ to the manager is a more effective instrument for mitigating hiring sabotage than awarding a positive share $\beta$ to the loser: when the principal can use a head start, the optimal tournament is winner-take-all, and the optimal head start $h^*$ and payout ratio $\alpha^*$ characterized in Proposition~\ref{prop:2} remain unchanged.

\subsection{Proofs for Appendix \ref{app:beta}}
\label{app:proofs-beta}

\begin{lemma}[Equilibrium uniqueness in the tournament subgame]
\label{lem:unique-continuation}
Under the assumptions of Proposition~\ref{prop:B1}, for every feasible sharing rule $(\alpha,\beta)$, head start $h$,
and hiring choice $\theta_n\in\Theta$, the tournament subgame
$\Gamma(\theta_n)$ has at most one pure-strategy Nash equilibrium.
\end{lemma}

\begin{proof}
If $\gamma=0$, winning is payoff-irrelevant, and the unique
equilibrium is
$(e_m,e_n)=(\alpha\theta_m,\alpha\theta_n)$. Hence, suppose
$\gamma>0$. The zero-effort profile cannot be an equilibrium
because $\alpha>0$ implies that at least one player has a strictly positive
marginal return to effort.

First, no equilibrium can satisfy
$e_m-e_n+h=-M/2$. At the support boundary, the manager's one-sided marginal returns to effort from the left and right are
\[
u_m'(e_m-) = D_-=\beta-\frac{e_m}{\theta_m},
\qquad
u_m'(e_m+) = D_+=\beta+\frac{\gamma (e_m+e_n)}{M}
          -\frac{e_m}{\theta_m}.
\]
If $e_m>0$, optimality would require
$D_-\geq0\geq D_+$, which is impossible because
$D_+>D_-$. If $e_m=0$, then $D_+>0$ (unless $e_n=\beta=0$, which is impossible because at least one player must have a strictly positive marginal return), which is also inconsistent
with optimality. Applying the same argument to the new hire rules
out an equilibrium at $e_m-e_n+h=M/2$.

Every pure-strategy equilibrium must therefore be in one of three
regimes. An equilibrium with an interior winning probability must
satisfy
\[
e_m^I
=
\frac{\theta_m(\tau M/2+\gamma h)}
     {M-2\gamma\theta_m},
\qquad
e_n^I
=
\frac{\theta_n(\tau M/2-\gamma h)}
     {M-2\gamma\theta_n}.
\]
These are the unique solutions to the interior first-order
conditions because $M>2\gamma\bar\theta$. 

If the manager loses with
probability one, equilibrium efforts must be
$(\beta\theta_m,\alpha\theta_n)$; if she wins with probability one,
they must be $(\alpha\theta_m,\beta\theta_n)$. Hence, there is at
most one equilibrium within each regime.

The two strict-corner profiles cannot both be equilibria. Their
support conditions would jointly imply
\[
M<\gamma(\theta_m+\theta_n)\leq2\bar\theta,
\]
contrary to $M>(1+\sqrt{2})\bar\theta$.

It remains to rule out coexistence of an interior equilibrium and a
strict-corner equilibrium. Suppose, toward a contradiction, that
both the interior equilibrium $(e_m^I,e_n^I)$ and the strict-corner equilibrium
$(\beta\theta_m,\alpha\theta_n)$ where the manager always loses exist. The latter requires
\[
\beta\theta_m-\alpha\theta_n+h<-\frac M2.
\]
Using the interior effort formulas, this condition implies
\[
e_n^I>\alpha\theta_n,
\qquad
e_m^I-\beta\theta_m
<
\frac{\gamma\tau\bar\theta^2}
     {M-2\gamma\bar\theta}
<M.
\]
The last inequality follows from
$M>(1+\sqrt2)\bar\theta$, since
\[M^2-2\gamma M\bar\theta-\gamma\tau\bar\theta^2
\geq M^2-2M\bar\theta-\bar\theta^2>0.\]

For every
$\tilde e_n\in[\alpha\theta_n,e_n^I]$, choosing
$\beta\theta_m$ leaves the manager certain to lose because
\[
\beta\theta_m-\tilde e_n+h
\leq
\beta\theta_m-\alpha\theta_n+h
<-\frac M2.
\]
By contrast, choosing $e_m^I$ gives the manager an interior winning
probability throughout this interval. This holds at
$\tilde e_n=e_n^I$ because $(e_m^I,e_n^I)$ is an interior
equilibrium. At the other endpoint,
\[
e_m^I-\alpha\theta_n+h
=
(e_m^I-\beta\theta_m)
+(\beta\theta_m-\alpha\theta_n+h)
<
M-\frac M2
=
\frac M2.
\]
Moreover, because $e_n^I>\alpha\theta_n$,
\[
e_m^I-\alpha\theta_n+h
>
e_m^I-e_n^I+h
>-\frac M2.
\]

Consider the manager's gain from choosing $e_m^I$ rather than
$\beta\theta_m$:
\[
\Delta(\widetilde e_n)
=
u_m(e_m^I,\widetilde e_n,\theta_m)
-
u_m(\beta\theta_m,\widetilde e_n,\theta_m).
\]
On $[\alpha\theta_n,e_n^I]$,  we have
\[
\Delta'(\widetilde e_n)=
\frac{\gamma}{M}
\left(
\frac{M}{2}+h-2\widetilde e_n
\right),
\]
which is negative for every $\widetilde e_n\geq\alpha\theta_n$ because $
\beta\theta_m-\alpha\theta_n+h<-\frac M2$ implies
\[
\frac{M}{2}+h-2\widetilde e_n
<
\alpha\theta_n-\beta\theta_m-2\widetilde e_n
\leq
-\alpha\theta_n-\beta\theta_m
\leq0.
\]
Hence, the manager's gain from choosing $e_m^I$ rather than
$\beta\theta_m$ is strictly decreasing in $\widetilde e_n$ on
$[\alpha\theta_n,e_n^I]$.

Because $(\beta\theta_m,\alpha\theta_n)$ is an equilibrium,
\[
\Delta(\alpha\theta_n)= u_m(e_m^I,\alpha\theta_n,\theta_m)
-
u_m(\beta\theta_m,\alpha\theta_n,\theta_m)
\leq0.
\]
Since $e_n^I>\alpha\theta_n$, strict monotonicity then gives
\[
 \Delta(e_n^I)= u_m(e_m^I,e_n^I,\theta_m)
-
u_m(\beta\theta_m,e_n^I,\theta_m)
<0,
\]
contradicting the optimality of $e_m^I$ at the interior equilibrium.

Interchanging the two players and replacing $h$ with $-h$ rules
out coexistence of the interior equilibrium and the strict-corner
equilibrium $(\alpha\theta_m,\beta\theta_n)$ where the manager always wins.

Thus, the tournament subgame has at most one pure-strategy Nash
equilibrium.
\end{proof}

\begin{lemma}[Equilibrium existence under winner-take-all]\label{lem:existence}
    Under the assumptions of Proposition~\ref{prop:B1}, for every $\theta_n\in\Theta$, the winner-take-all contract $(\alpha^*,\beta^*,h^*)$, where 
    \[
     \alpha^*
     =
     \left(
     1+\sqrt{1-\frac{2\theta_m}{M}}
     \right)^{-1}, \quad \beta^* = 0, \quad h^* = 2\alpha^*\bar\theta-\frac{M}{2},
     \]
    induces a unique pure-strategy Nash equilibrium in the tournament subgame $\Gamma(\theta_n)$, as characterized by the first-order conditions~\eqref{FOC1} and~\eqref{FOC2}:
    \begin{align*} 
        e_m^*(\theta_n) &=  \frac{M/2+ h^*}{M/\alpha^*\theta_m -2}=
\frac{2(\alpha^*)^2\bar\theta\theta_m}
     {M-2\alpha^*\theta_m},\\
        e_n^*(\theta_n) &=  \frac{M/2-h^*}{M/\alpha^* \theta_n -2}=
\frac{\alpha^*(M-2\alpha^*\bar\theta)\theta_n}
     {M-2\alpha^*\theta_n}.
    \end{align*}
\end{lemma}

\begin{proof}    
Fix any $\theta_n\in\Theta$. Since $\theta_m\leq\bar\theta$ and
$M>(1+\sqrt{2})\bar\theta \ge (1+\sqrt{2})\theta_m$, we have
\[
\alpha^*= \left(1 + \sqrt{1-2\frac{\theta_m}{M}}\right)^{-1} <\frac{1}{\sqrt{2}},
\qquad
M>(1+2\alpha^*)\bar\theta,
\qquad
M>2\alpha^*(1+\alpha^*)\bar\theta.
\]
Therefore, \(e_m^*<\bar\theta<M-2\alpha^*\bar\theta\) and \(e_n^*(\theta_n)\le\alpha^*\theta_n\).

For any $e_m\geq0$, holding $e_n=e_n^*(\theta_n)$,
\[
\frac{M}{2}+e_m-e_n^*(\theta_n)+h^*
=
\left(
\frac{M}{2}+h^*-e_n^*(\theta_n)
\right)+e_m
>0.
\]
Thus, the manager's winning probability is interior or equal to one.
Similarly, for any $e_n\geq0$, holding $e_m=e_m^*$,
\[
\frac{M}{2}+e_n-e_m^*-h^*
=
\left(
\frac{M}{2}-h^*-e_m^*
\right)+e_n
>0,
\]
so the new hire's winning probability is also interior or equal to
one.

Now we show that the payoff function of each player is strictly concave in their own effort.
For $i\in\{m,n\}$,  when the winning probability is interior, we have
\[
\frac{\partial^2 u_i}{\partial e_i^2} = \frac{2\alpha^*}{M}-\frac{1}{\theta_i}<0.
\]
When the winning probability is equal to one, we have
$\frac{\partial^2 u_i}{\partial e_i^2} =-\frac{1}{\theta_i}<0$.
At the
boundary between these regions, the marginal payoff decreases because
additional effort does not increase the winning probability. Hence,
against the displayed opponent effort, each player's payoff is
strictly concave over the entire domain $e_i\geq0$.

The previous bounds also imply
\[
e_m^*-e_n^*(\theta_n)+h^*+\frac{M}{2}
= e_m^*-e_n^*(\theta_n) +2\alpha^*\bar\theta
\in(0,M),
\]
so $(e_m^*, e_n^*(\theta_n))$ has an interior winning probability.
Strict concavity of the payoff functions therefore implies that 
the effort levels are the unique best responses to each other. Hence,
$(e_m^*, e_n^*(\theta_n))$ constitutes a pure-strategy Nash equilibrium of the tournament subgame $\Gamma(\theta_n)$.
By Lemma~\ref{lem:unique-continuation}, it is the unique equilibrium. 

\end{proof}

\begin{proof}[Proof of Proposition~\ref{prop:B1}]
Fix a contract $(\alpha,\beta,h)$ for which $\Gamma(\theta_n)$ has a
pure-strategy Nash equilibrium for every $\theta_n\in\Theta$. By
Lemma~\ref{lem:unique-continuation}, this equilibrium is unique.

Because each agent's expected revenue share is at most $\alpha\leq1$ and the two agents' shares sum to $\tau=\alpha+\beta\leq1$, the proof of Lemma~\ref{lem:bounded} applies unchanged under a general sharing rule. Hence, equilibrium efforts are uniformly bounded by $\bar e$.
Continuity of payoffs and uniqueness therefore imply that the
equilibrium effort profile is continuous in $\theta_n$.

\textbf{Case 1.}
Suppose that $\gamma>0$ and that the equilibrium following the
manager's hiring choice $\theta_n^*$ has an interior winning
probability. Continuity of the equilibrium effort profile implies that
the winning probability remains interior for all $\theta_n$
sufficiently close to $\theta_n^*$. Hence, the interior effort formulas
and the hiring first-order or endpoint condition are valid locally.

Define
\[
\hat\theta
=
\frac{\tau M/2+\gamma h}{2\gamma\tau}.
\]
The tournament-stage first-order conditions give
\[
e_m
=
\frac{2\gamma\tau\hat\theta\,\theta_m}
     {M-2\gamma\theta_m},
\qquad
e_n(\theta_n)
=
\frac{\tau(M-2\gamma\hat\theta)\theta_n}
     {M-2\gamma\theta_n}.
\]
By the envelope theorem,
\[
 \frac{d u_m}{d\theta_n}
 =
 \left[
 \beta+\gamma G(e_m-e_n+h)-\frac{\gamma(e_m+e_n)}{M}
 \right]e_n'(\theta_n)
 =  
\frac{2\gamma\tau(\hat\theta-\theta_n)}
     {M-2\gamma\theta_n} e_n'(\theta_n),
\]
where \(e_n'(\theta_n)
>0\). Thus, the derivative has the same sign as
$\hat\theta-\theta_n$.

If $\theta_n^*\in(\underline\theta,\bar\theta)$, hiring optimality
implies $\hat\theta=\theta_n^*$. If
$\theta_n^*=\bar\theta$, it implies
$\hat\theta\geq\bar\theta$, and if
$\theta_n^*=\underline\theta$, it implies
$\hat\theta\leq\underline\theta$. In all three cases,
using $\underline\theta\leq\theta_m\leq\bar\theta$, we obtain
\begin{equation}
e_m+e_n
\leq
\frac{\tau\bar\theta M}
     {M-2\gamma\theta_m}.
\end{equation}

\textbf{Case 2.}    
Next, suppose that the equilibrium lies at a strict probability
corner, so that
\[
e_m-e_n+h<-\frac M2
\quad\text{or}\quad
e_m-e_n+h>\frac M2.
\] If the manager loses with probability one, equilibrium
efforts are $(e_m,e_n)=(\beta\theta_m,\alpha\theta_n^*)$.
     If she wins with probability one,  equilibrium
efforts are $(e_m,e_n)=(\alpha\theta_m,\beta\theta_n^*)$.
Thus, in either case,
\[
e_m+e_n\leq\tau\bar\theta
\leq
\frac{\tau\bar\theta M}
     {M-2\gamma\theta_m}.
\]

Moreover, as shown in the proof of Lemma~\ref{lem:unique-continuation}, no
pure-strategy Nash equilibrium can satisfy
$|e_m-e_n+h|= M/2$.

\textbf{Case 3.}
Finally, if $\gamma=0$, equilibrium efforts are $(e_m,e_n) = (\frac{\tau\theta_m}{2},\frac{\tau\theta_n^*}{2})$.
Hence, $e_m+e_n\leq\tau\bar\theta$.

Combining all these cases above gives
\[
e_m+e_n
\leq
\frac{\tau\bar\theta M}
     {M-2\gamma\theta_m}
\leq
\frac{\tau\bar\theta M}
     {M-2\tau\theta_m}.
\]
Therefore,
\begin{equation} \label{eqn:barpi}
   \Pi(\alpha,\beta,h)
\leq
(1-\tau)
\frac{\tau\bar\theta M}
     {M-2\tau\theta_m}
\equiv\overline\Pi(\tau)  
\end{equation}
and the inequality is strict if $\gamma<\tau$.

Direct maximization of $\overline\Pi(\tau)$ gives the unique maximizer
\[
\tau^*
=
\left(1+\sqrt{1-\frac{2\theta_m}{M}}\right)^{-1}.
\]
Consider the feasible contract
\[
 \alpha^*=\gamma^*=\tau^*,  \quad
 \beta^*=0,
 \quad
 h^*= \bar h(\alpha^*)  = 2\alpha^*\bar\theta-\frac M2
\]
Lemma~\ref{lem:existence} shows that this contract induces a unique
pure-strategy Nash equilibrium in every tournament subgame. Moreover,
by the definition of $\bar h(\alpha^*)$, the manager hires
$\theta_n^*=\bar\theta$. 
At this hiring choice,
\[
e_n^*=\alpha^*\bar\theta,
\quad
e_m^*
=
\frac{2(\alpha^*)^2\bar\theta\theta_m}
     {M-2\alpha^*\theta_m},
\]
and hence the contract attains $\overline\Pi(\alpha^*)$.

Because the inequality in~\eqref{eqn:barpi} is strict whenever
$\gamma<\tau$, and $\gamma<\tau$ whenever $\beta>0$, the unique optimal sharing rule satisfies $\beta^*=0$ and $\alpha^*=\tau^*$.
Finally, conditional on $\alpha^*$,
Proposition~\ref{prop:1} implies that
$h^*=\bar h(\alpha^*)$ is the unique optimal head start if
$\theta_m<\bar\theta$.
Hence,
$(\alpha^*,\beta^*,h^*)$ is the unique optimal contract if
$\theta_m<\bar\theta$.
\end{proof}

\section{A Two-Period Model}
\label{twoperiod}
\subsection{Setup}

In the main text, we show that the optimal head start always ensures that the manager hires the strongest candidate---that is, it eliminates hiring sabotage.
To examine the robustness of this result, we now extend the model to two periods to account for future profitability of the firm by assuming that the tournament winner is promoted (or retained) in the next period. 
As is common in practice, the firm uses the same rule for both compensation and promotion for the sake of simplicity and transparency.
Thus, %
the principal now incorporates the winner's ability into her objectives, as promoting a higher-ability agent from the tournament is vital for future profitability.

Assume the promoted agent receives a constant continuation payoff $\tilde v>0$. 
After the first-period tournament, the principal receives a continuation payoff $\tilde V(\theta)>0$ if the agent who is promoted (i.e., who wins in the first period) has ability $\theta$. Assume that $\tilde V(\theta)$ is increasing.

Alternatively, the continuation payoff $\tilde V$ can be interpreted as capturing the cost of providing head starts in the one-period model, which is incurred when a weaker player wins the tournament. When a weaker player wins the tournament, the principal's reputation may suffer or its future profit may decline.

\subsection{Agent's Problem}

For simplicity, we focus on the uniform--quadratic case: $c(e)=e^2/2$ and $G\sim \operatorname{Unif}[-M/2,M/2]$.
Let $\delta\in(0,1)$ denote the discount factor.
Given the manager’s hiring decision \(\theta_n\), the expected payoffs of \(m\) and \(n\) in the subgame are given by
\begin{align} \label{eq:30}
& u_m%
    = \alpha \cdot (e_m + e_n) G(e_m - e_n+h) - \frac{e_m^2}{2\theta_m} 
    + \delta  G(e_m - e_n+h)  \cdot \tilde v,\\
   & u_n%
    = \alpha \cdot (e_m + e_n)G(e_n - e_m-h)  \, - \frac{e_n^2}{2\theta_n} 
    \; + \delta G(e_n - e_m-h) \cdot  \tilde v.
\end{align}
Analogous to the one-period case, the equilibrium efforts in the tournament stage given $\theta_n$ are
\begin{align}
e_m(\theta_n)
        = \frac{(M/2+h) + \delta \tilde v / \alpha}{M/\alpha \theta_m - 2},\\
            e_n(\theta_n)
        = \frac{(M/2-h) + \delta \tilde v/ \alpha}{M/\alpha \theta_n - 2}.
\end{align}
At the hiring stage in the first period, the manager chooses $\theta_n^*$ according to the first-order condition
\begin{equation}
    \begin{aligned}\label{eq18}
        \frac{d u_m}{d \theta_n}=
        e_n'(\theta_n) [ \alpha \cdot (G(e_m -e_n +h)-g(e_m - e_n +h)(e_m +e_n)) - \delta  g(e_m -e_n +h) \cdot \tilde v]=0. %
    \end{aligned}
\end{equation}
Therefore, in the SPNE,
\begin{align} \label{eq20}
    & \theta_n^* = \min\left(\frac{ (M/2 +h) -\delta \tilde v /\alpha}{2\alpha}, \bar{\theta}\right),\\
    & e_n^* \equiv e_n(\theta_n^*) = \min \left( \frac{ (M/2 +h) -\delta \tilde v /\alpha}{2},\frac{(M/2-h) + \delta \tilde v / \alpha}{M/\alpha \bar\theta - 2} \right),\\
    & e_m^* = \frac{(M/2+h) + \delta \tilde v / \alpha}{M/\alpha \theta_m - 2}.
\end{align}
The amount of head start necessary to induce the manager to hire the strongest candidate is
\begin{equation}
\bar  h(\alpha)  = 
   2 \alpha \bar\theta- \frac{M}{2} + \delta \tilde v/\alpha  >   2 \alpha  \bar\theta -  \frac{M}{2},  %
\end{equation}
which is larger than that in the one-period model where $\bar h (\alpha)= 2 \alpha  \bar\theta - M/2$ because of the continuation value.
We assume that $M$ is sufficiently large so that the winning probabilities lie in $(0,1)$ when the head start is $\bar h(\alpha)$.

As we shall see below, the \emph{encouragement}, \emph{discouragement}, and \emph{hiring effects} observed in the one-period model continue to operate, and Lemma~\ref{lemma:1} still holds.

\subsection{Principal's Problem}
The principal's problem is 
\begin{equation}
    \begin{aligned}
        \max_{\alpha\in[0,1],\;h\in [-\bar H ,\bar H]}  \tilde \Pi(\alpha,h) =    (1-\alpha) (e_n^* +e_m^* )  + \delta \tilde V(\theta_{m})G(e_m^* -e_n^* +h)  \\ + \delta \tilde V(\theta_n^*(h) ) [1-G(e_m^* -e_n^* +h) ].  
    \end{aligned}
\end{equation}

Using two-step maximization again, the first-order condition for $h$ is
\begin{align}\label{FOC}
    &\begin{aligned}
    \frac{d\tilde \Pi}{d h} = \underbrace{ (1-\alpha ) (\frac{d e_n }{d h}+\frac{d e_m }{d h})  }_{ \text{first-period effects ($>0$ if $\theta_n^* <\bar\theta$)}}
        +  \delta \underbrace{  (\tilde V(\theta_m)-\tilde V(\theta_n^*(h)))\left(\frac{d e_m }{d h}-\frac{d e_n }{d h}+1\right) g(e_m^* -e_n^* +h)}_{\text{Succession effect}} \\
      ~  + \delta \underbrace{ \tilde V'(\theta_n^*(h))[1-G(e_m^* -e_n^*+h) ]\theta_n^{*\prime}(h)}_{\text{Extended hiring effect}}. %
    \end{aligned} 
\end{align}
For simplicity, assume that the principal's continuation payoff is given by $\tilde V(\theta)= k \theta+ b$ with $k, b>0$, which captures the succession concern.

\paragraph{Succession and Hiring Effects.}
Until the strongest candidate is hired (when $\theta_n^* <\bar\theta$), the head start increases the manager's effort because of the \emph{encouragement effect} ($de_m^*/dh>0$) and also increases the new hire's effort as the \emph{hiring effect} still dominates the \emph{discouragement effect} ($de_n^*/dh=1/2$ as in equation~\eqref{eq: DECE}).  %
Thus, they aggregate to a positive effect on the profit as long as $\theta_n^*<\bar\theta$ as in Lemma~\ref{lemma:1}.
However, in a model with succession concerns, a head start has two additional effects:
\begin{enumerate}\setcounter{enumi}{3}
    \item Extended hiring effect: $[1-G(e_m -e_n+h) ] \delta k \theta_n^{*\prime}(h) \geq0$.
    \item Succession effect: $\delta k ( \theta_m -\theta_n^*(h))(\frac{d e_m }{d h}-\frac{d e_n }{d h}+1) g(e_m^*-e_n^*+h)$, which is nonpositive if and only if $\theta_n^*(h)\geq\theta_m$.
\end{enumerate}

As observed in the one-period model, the head start partially insulates the hiring manager from competition and leads her to hire a higher-ability agent than she otherwise would, who could be promoted for the next period with positive probability.
This \emph{extended hiring effect} has a nonnegative impact on the principal's continuation payoff. Conversely, the head start also increases the probability of promoting the manager for the future period. Such a \emph{succession effect} can be detrimental to future profit if and only if the manager's ability is lower than the new hire's.

\paragraph{Alternative Interpretation.}
It is worth noting that when $\theta_n^*(h)\geq\theta_m$, the term $\delta (\tilde V(\theta_m) - \tilde{V}(\theta_n^*(h))) G(e_m -e_n+h)<0$ can also be viewed as the cost of providing a head start, which arises when the weaker player wins the tournament, as it potentially damages the firm's reputation and lowers employee morale.
Therefore, the succession effect can be interpreted as the reputation or morale effect.

\begin{proposition}
In the two-period model, the optimal head start may allow for hiring sabotage in equilibrium (i.e., $\theta_n^*(h^*)<\bar\theta$) when the future stake is large ($\delta k>0$).
However, the new hire is always strictly better than the manager when sabotage arises in equilibrium.
\end{proposition}
\begin{proof}
    To show the existence of hiring sabotage in equilibrium with the optimal head start, we show $\theta_n^*(h^*)<\bar\theta$.
    When $\delta = 0.85$, $k = 16.4$, $M = 2.55$, $\tilde v =1$, $\theta_m \in(0,0.6)$, and $\bar\theta=1$, we have $\theta_n^*<\bar\theta$.%

    Now we prove that $\theta_n^*(h^*)<\bar\theta$ implies $\theta_n^*(h^*)>\theta_m$ by contradiction.
    Suppose that $\theta_n^*(h^*)<\bar\theta$ and $\theta_n^*(h^*) \leq \theta_m$. 
    Then, the succession effect is nonnegative, and the first-period effects are positive.
    Thus, equation~\eqref{FOC} implies $d\tilde \Pi/ d h>0$, so increasing $\theta_n^*$ by increasing $h$ would strictly increase $\tilde \Pi$, which contradicts the optimality of $h^*$.    
    \end{proof}

\begin{figure}[htbp]
\centering
\includegraphics[width=0.95\textwidth]{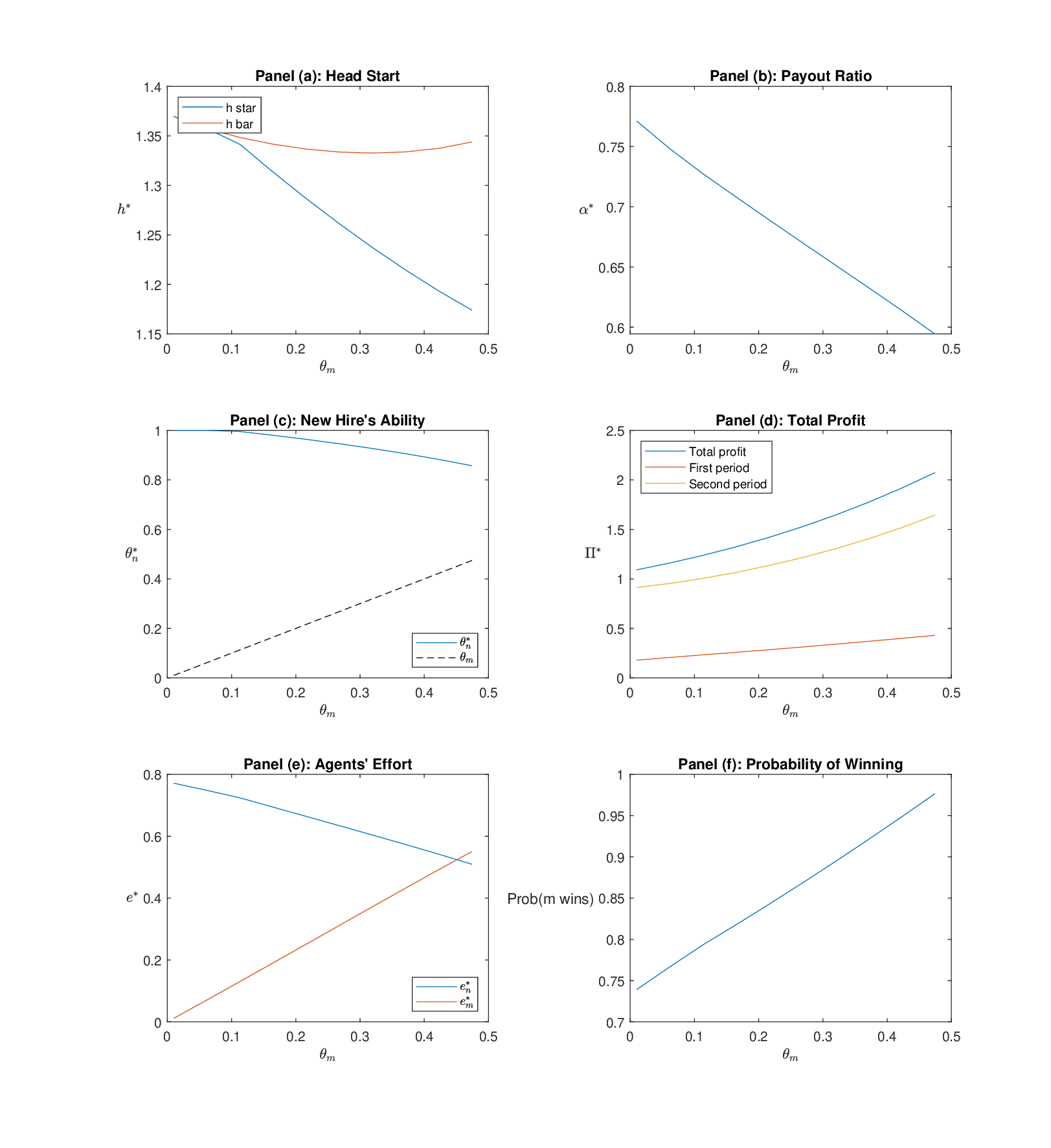}
\caption{Hiring sabotage arises due to continuation payoffs}
\label{fig:1}
\end{figure}

In sharp contrast to the benchmark one-period model, in the two-period model where the winner's ability matters, the optimal head start may allow for hiring sabotage in equilibrium. 
This is driven by the \emph{succession effect} of the head start.
At the level of the head start that eliminates hiring sabotage (i.e., $h=\bar h$), if the (negative) succession effect dominates the sum of the aggregate positive effect on the first-period profit and the extended hiring effect on the continuation profit, the principal will lower the head start to increase the profit, which creates scope for hiring sabotage. This only happens when the manager's ability is lower than the new hire's because the succession effect would be positive otherwise.
Therefore, regardless of the existence of hiring sabotage, the optimal contract ensures that the new hire is always of higher ability than the manager (i.e., $\theta_n^*\in(\theta_m,\bar\theta]$).
\footnote{If the optimal contract allows for hiring sabotage, the new hire must be of higher ability than the manager. If it does not, then the new hire is $\theta_n=\bar\theta$. In either case,  $\theta_n^*\in(\theta_m,\bar\theta]$. }

Figure~\ref{fig:1} illustrates the optimal head start $h^*$, optimal payout ratio $\alpha^*$, and the new hire's ability $\theta_n^*$ as a function of the manager's ability $\theta_m\in(0,0.5)$ when $\delta = 0.85$, $\bar\theta=1$, $k = 4$, $M = 2.55$, and $\tilde v =1$.%
\footnote{For $\theta_m>0.5$, the probability of winning is 1, as shown in panel (f), so a pure-strategy equilibrium does not exist.}
It can be seen in panel (a) that the optimal head start $h^*$ (blue line) is smaller than the sabotage-free level $\bar h$ (red line), which allows for hiring sabotage---the new hire's ability $\theta_n^*$ (blue line), as shown in panel (c), is lower than $\bar \theta=1$. However, it is still higher than the manager's ability ($45^{\circ}$ dashed line).

Moreover, it can be deduced from panels (a) and (b) that an increase in the manager's ability $\theta_m$ \emph{decreases} both the optimal head start $h^*$ and the optimal payout ratio $\alpha^*$ (blue line), in contrast to the one-period model. 
This is because the agents now have career incentives (continuation payoffs $v$) in addition to the first-period tournament prize incentives. As the manager's ability increases, the career incentive encourages her to invest more effort, thereby allowing the principal to lower the optimal payout ratio  $\alpha^*$, which in turn lowers the optimal head start $h^*$.
As the manager's ability $\theta_m$ increases, the decrease in the optimal head start $h^*$ and the optimal payout ratio $\alpha^*$ have opposite effects on hiring sabotage---the former exacerbates it and the latter mitigates it. 
In combination, the former dominates the latter, and they jointly lead to a decrease in the new hire's ability $\theta_n^*$, as shown in panel (c).
Nevertheless, the firm's total profit is still increasing in the manager's ability $\theta_m$, as shown in panel (d), because the direct effect of a better manager on the profit outweighs the indirect effect due to the increase in hiring sabotage.

\end{appendices}

\section*{Declaration of generative AI and AI-assisted technologies in the manuscript preparation process}

During the preparation of this work, the authors used ChatGPT and Claude to assist with copyediting, conducting searches for related literature, and checking mathematical proofs. After using these tools, the authors reviewed and edited the content as needed and take full responsibility for the content of the published article.

\singlespacing
\urlstyle{same}
\bibliographystyle{jpe}
\bibliography{nelly}

\begin{thebibliography}{54}
\newcommand{\enquote}[1]{``#1''}
\providecommand{\natexlab}[1]{#1}
\providecommand{\url}[1]{\texttt{#1}}
\providecommand{\urlprefix}{URL }

\bibitem[{Baye and Hoppe(2003)}]{BayeHoppe2003}
Baye, Michael~R. and Heidrun~C. Hoppe. 2003.
\newblock \enquote{The Strategic Equivalence of Rent-Seeking, Innovation, and Patent-Race Games.}
\newblock \emph{Games and Economic Behavior} 44~(2):217--226.

\bibitem[{Brown and Chowdhury(2017)}]{brown2017hidden}
Brown, Alasdair and Subhasish~M Chowdhury. 2017.
\newblock \enquote{The Hidden Perils of Affirmative Action: Sabotage in Handicap Contests.}
\newblock \emph{Journal of Economic Behavior \& Organization} 133:273--284.

\bibitem[{Brown and Minor(2014)}]{BrownMinor2014}
Brown, Jennifer and Dylan~B. Minor. 2014.
\newblock \enquote{Selecting the {{Best}}? {{Spillover}} and {{Shadows}} in {{Elimination Tournaments}}.}
\newblock \emph{Management Science} 60~(12):3087--3102.

\bibitem[{Carmichael(1988)}]{carmichael1988incentives}
Carmichael, H~Lorne. 1988.
\newblock \enquote{Incentives in Academics: Why Is There Tenure?}
\newblock \emph{Journal of Political Economy} 96~(3):453--472.

\bibitem[{Chen(2024)}]{Chen2024}
Chen, Joanne. 2024.
\newblock \enquote{Optimal {{Managerial Authority}}.}
\newblock Available at SSRN: \url{https://papers.ssrn.com/abstract=4840824}.

\bibitem[{Chen(2003)}]{chen2003sabotage}
Chen, Kong-Pin. 2003.
\newblock \enquote{Sabotage in Promotion Tournaments.}
\newblock \emph{Journal of Law, Economics, and Organization} 19~(1):119--140.

\bibitem[{Chen, Jungbauer, and Wang(2023)}]{ChenJungbauerWang2023}
Chen, Yi, Thomas Jungbauer, and Zhe Wang. 2023.
\newblock \enquote{The Strategic Decentralization of Recruiting.}
\newblock \emph{Journal of Economic Theory} 209:105639.

\bibitem[{Chowdhury and G{\"u}rtler(2015)}]{chowdhury2015sabotage}
Chowdhury, Subhasish~M and Oliver G{\"u}rtler. 2015.
\newblock \enquote{Sabotage in Contests: A Survey.}
\newblock \emph{Public Choice} 164:135--155.

\bibitem[{Chung(1996)}]{Chung1996}
Chung, Tai-Yeong. 1996.
\newblock \enquote{Rent-{{Seeking Contest When}} the {{Prize Increases}} with {{Aggregate Efforts}}.}
\newblock \emph{Public Choice} 87~(1/2):55--66.

\bibitem[{Clark and Riis(2001)}]{clark2001rank}
Clark, Derek~J and Christian Riis. 2001.
\newblock \enquote{Rank-Order Tournaments and Selection.}
\newblock \emph{Journal of Economics} 73:167--191.

\bibitem[{Dai and Toikka(2022)}]{DaiToikka2022}
Dai, Tianjiao and Juuso Toikka. 2022.
\newblock \enquote{Robust {{Incentives}} for {{Teams}}.}
\newblock \emph{Econometrica} 90~(4):1583--1613.

\bibitem[{Danilov, Harbring, and Irlenbusch(2019)}]{DanilovHarbringIrlenbusch2019}
Danilov, Anastasia, Christine Harbring, and Bernd Irlenbusch. 2019.
\newblock \enquote{Helping under a Combination of Team and Tournament Incentives.}
\newblock \emph{Journal of Economic Behavior \& Organization} 162:120--135.

\bibitem[{Dasaratha, Golub, and Shah(2024)}]{DasarathaGolubShah2024}
Dasaratha, Krishna, Benjamin Golub, and Anant Shah. 2024.
\newblock \enquote{Incentive {{Design}} with {{Spillovers}}.}
\newblock {arXiv}:2411.08026.

\bibitem[{Deller and Sandino(2020)}]{DellerSandino2020}
Deller, Carolyn and Tatiana Sandino. 2020.
\newblock \enquote{Who {{Should Select New Employees}}, {{Headquarters}} or the {{Unit Manager}}? {{Consequences}} of {{Centralizing Hiring}} at a {{Retail Chain}}.}
\newblock \emph{The Accounting Review} 95~(4):173--198.

\bibitem[{Dessein(2002)}]{Dessein2002}
Dessein, Wouter. 2002.
\newblock \enquote{Authority and {{Communication}} in {{Organizations}}.}
\newblock \emph{The Review of Economic Studies} 69~(4):811--838.

\bibitem[{Deutscher et~al.(2013)Deutscher, Frick, G{\"u}rtler, and Prinz}]{deutscher2013sabotage}
Deutscher, Christian, Bernd Frick, Oliver G{\"u}rtler, and Joachim Prinz. 2013.
\newblock \enquote{Sabotage in Tournaments with Heterogeneous Contestants: Empirical Evidence from the Soccer Pitch.}
\newblock \emph{The Scandinavian Journal of Economics} 115~(4):1138--1157.

\bibitem[{Drugov and Ryvkin(2017)}]{DrugovRyvkin2017}
Drugov, Mikhail and Dmitry Ryvkin. 2017.
\newblock \enquote{Biased Contests for Symmetric Players.}
\newblock \emph{Games and Economic Behavior} 103:116--144.

\bibitem[{Drugov and Ryvkin(2020)}]{DrugovRyvkin2020}
---{}---{}---. 2020.
\newblock \enquote{How Noise Affects Effort in Tournaments.}
\newblock \emph{Journal of Economic Theory} 188:105065.

\bibitem[{Drugov and Ryvkin(2022)}]{drugov2022hunting}
---{}---{}---. 2022.
\newblock \enquote{Hunting for the Discouragement Effect in Contests.}
\newblock \emph{Review of Economic Design} .

\bibitem[{Ederer(2010)}]{Ederer2010}
Ederer, Florian. 2010.
\newblock \enquote{Feedback and {{Motivation}} in {{Dynamic Tournaments}}.}
\newblock \emph{Journal of Economics \& Management Strategy} 19~(3):733--769.

\bibitem[{Frankel(2021)}]{Frankel2021}
Frankel, Alex. 2021.
\newblock \enquote{Selecting {{Applicants}}.}
\newblock \emph{Econometrica} 89~(2):615--645.

\bibitem[{Friebel and Raith(2004)}]{FriebelRaith2004}
Friebel, Guido and Michael Raith. 2004.
\newblock \enquote{Abuse of {{Authority}} and {{Hierarchical Communication}}.}
\newblock \emph{The RAND Journal of Economics} 35~(2):224--244.

\bibitem[{Fu and Wu(2020)}]{fu2020optimal}
Fu, Qiang and Zenan Wu. 2020.
\newblock \enquote{On the Optimal Design of Biased Contests.}
\newblock \emph{Theoretical Economics} 15~(4):1435--1470.

\bibitem[{Gershkov, Li, and Schweinzer(2009)}]{GershkovLiSchweinzer2009}
Gershkov, Alex, Jianpei Li, and Paul Schweinzer. 2009.
\newblock \enquote{Efficient {{Tournaments}} within {{Teams}}.}
\newblock \emph{The RAND Journal of Economics} 40~(1):103--119.

\bibitem[{Gl{\"o}kler, Pull, and Stadler(2022)}]{GloklerPullStadler2022}
Gl{\"o}kler, Thomas, Kerstin Pull, and Manfred Stadler. 2022.
\newblock \enquote{Do {{Output-Dependent Prizes Alleviate}} the {{Sabotage Problem}} in {{Tournaments}}?}
\newblock \emph{Games} 13~(5):65.

\bibitem[{G{\"u}rtler and M{\"u}nster(2010)}]{gurtler2010sabotage}
G{\"u}rtler, Oliver and Johannes M{\"u}nster. 2010.
\newblock \enquote{Sabotage in Dynamic Tournaments.}
\newblock \emph{Journal of Mathematical Economics} 46~(2):179--190.

\bibitem[{G{\"u}th et~al.(2016)G{\"u}th, Lev{\'i}nsk{\'y}, Pull, and Weisel}]{GuthLevinskyPull2016}
G{\"u}th, Werner, Ren{\'e} Lev{\'i}nsk{\'y}, Kerstin Pull, and Ori Weisel. 2016.
\newblock \enquote{Tournaments and Piece Rates Revisited: A Theoretical and Experimental Study of Output-Dependent Prize Tournaments.}
\newblock \emph{Review of Economic Design} 20~(1):69--88.

\bibitem[{Haegele(2026)}]{Haegele2026}
Haegele, Ingrid. 2026.
\newblock \enquote{Talent {{Hoarding}} in {{Organizations}}.}
\newblock \emph{American Economic Review} 116~(8):3110--3151.

\bibitem[{Harbring and Irlenbusch(2008)}]{HarbringIrlenbusch2008}
Harbring, Christine and Bernd Irlenbusch. 2008.
\newblock \enquote{How Many Winners Are Good to Have?: {{On}} Tournaments with Sabotage.}
\newblock \emph{Journal of Economic Behavior \& Organization} 65~(3):682--702.

\bibitem[{Holmstrom(1982)}]{Holmstrom1982}
Holmstrom, Bengt. 1982.
\newblock \enquote{Moral {{Hazard}} in {{Teams}}.}
\newblock \emph{The Bell Journal of Economics} 13~(2):324--340.

\bibitem[{Hvide and Kristiansen(2003)}]{hvide2003risk}
Hvide, Hans~K and Eirik~G Kristiansen. 2003.
\newblock \enquote{Risk Taking in Selection Contests.}
\newblock \emph{Games and Economic Behavior} 42~(1):172--179.

\bibitem[{Kawasaki(2015)}]{Kawasaki2015}
Kawasaki, Guy. 2015.
\newblock \emph{The {{Art}} of the {{Start}} 2.0: {{The Time-Tested}}, {{Battle-Hardened Guide}} for {{Anyone Starting Anything}}}.
\newblock New York, NY: Penguin.

\bibitem[{Konrad(2002)}]{Konrad2002}
Konrad, Kai~A. 2002.
\newblock \enquote{Investment in the Absence of Property Rights; the Role of Incumbency Advantages.}
\newblock \emph{European Economic Review} 46~(8):1521--1537.

\bibitem[{Konrad(2009)}]{Konrad2009}
---{}---{}---. 2009.
\newblock \emph{Strategy and {{Dynamics}} in {{Contests}}}.
\newblock London {{School}} of {{Economics Perspectives}} in {{Economic Analysis}}. Oxford, New York: Oxford University Press.

\bibitem[{Kr{\"a}kel(2005)}]{krakel2005helping}
Kr{\"a}kel, Matthias. 2005.
\newblock \enquote{Helping and sabotaging in tournaments.}
\newblock \emph{International Game Theory Review} 7~(02):211--228.

\bibitem[{Lazear(1989)}]{lazear1989pay}
Lazear, Edward~P. 1989.
\newblock \enquote{Pay Equality and Industrial Politics.}
\newblock \emph{Journal of Political Economy} 97~(3):561--580.

\bibitem[{Lazear(1995)}]{lazear1995personnel}
---{}---{}---. 1995.
\newblock \emph{Personnel Economics}.
\newblock MIT Press.

\bibitem[{Lazear and Rosen(1981)}]{lazear1981rank}
Lazear, Edward~P and Sherwin Rosen. 1981.
\newblock \enquote{Rank-Order Tournaments as Optimum Labor Contracts.}
\newblock \emph{Journal of Political Economy} 89~(5):841--864.

\bibitem[{Morgan, Tumlinson, and V{\'a}rdy(2022)}]{MorganTumlinsonVardy2022}
Morgan, John, Justin Tumlinson, and Felix V{\'a}rdy. 2022.
\newblock \enquote{The Limits of Meritocracy.}
\newblock \emph{Journal of Economic Theory} 201:105414.

\bibitem[{M{\"u}nster(2007)}]{munster2007selection}
M{\"u}nster, Johannes. 2007.
\newblock \enquote{Selection Tournaments, Sabotage, and Participation.}
\newblock \emph{Journal of Economics \& Management Strategy} 16~(4):943--970.

\bibitem[{Nalebuff and Stiglitz(1983)}]{nalebuff1983prizes}
Nalebuff, Barry~J and Joseph~E Stiglitz. 1983.
\newblock \enquote{Prizes and Incentives: Towards a General Theory of Compensation and Competition.}
\newblock \emph{The Bell Journal of Economics} :21--43.

\bibitem[{O'Keeffe, Viscusi, and Zeckhauser(1984)}]{o1984economic}
O'Keeffe, Mary, W~Kip Viscusi, and Richard~J Zeckhauser. 1984.
\newblock \enquote{Economic Contests: Comparative Reward Schemes.}
\newblock \emph{Journal of Labor Economics} 2~(1):27--56.

\bibitem[{Purkayastha(1998)}]{Purkayastha1998}
Purkayastha, Sumitra. 1998.
\newblock \enquote{Simple Proofs of Two Results on Convolutions of Unimodal Distributions.}
\newblock \emph{Statistics \& Probability Letters} 39~(2):97--100.

\bibitem[{Ryvkin and Drugov(2020)}]{RyvkinDrugov2020}
Ryvkin, Dmitry and Mikhail Drugov. 2020.
\newblock \enquote{The Shape of Luck and Competition in Winner-take-all Tournaments.}
\newblock \emph{Theoretical Economics} 15~(4):1587--1626.

\bibitem[{Ryvkin and Ortmann(2008)}]{ryvkin2008predictive}
Ryvkin, Dmitry and Andreas Ortmann. 2008.
\newblock \enquote{The Predictive Power of Three Prominent Tournament Formats.}
\newblock \emph{Management Science} 54~(3):492--504.

\bibitem[{Sengupta(2004)}]{Sengupta2004}
Sengupta, Sarbajit. 2004.
\newblock \enquote{Delegating Recruitment under Asymmetric Information.}
\newblock \emph{International Journal of Industrial Organization} 22~(8):1327--1347.

\bibitem[{Shaked and Shanthikumar(2007)}]{ShakedShanthikumar2007}
Shaked, Moshe and J.~George Shanthikumar. 2007.
\newblock \emph{Stochastic {{Orders}}}.
\newblock Springer Science \& Business Media.

\bibitem[{Siegel(2014)}]{Siegel2014}
Siegel, Ron. 2014.
\newblock \enquote{Asymmetric {{Contests}} with {{Head Starts}} and {{Nonmonotonic Costs}}.}
\newblock \emph{American Economic Journal: Microeconomics} 6~(3):59--105.

\bibitem[{Skaperdas and Grofman(1995)}]{skaperdas1995modeling}
Skaperdas, Stergios and Bernard Grofman. 1995.
\newblock \enquote{Modeling Negative Campaigning.}
\newblock \emph{American Political Science Review} 89~(1):49--61.

\bibitem[{Sullivan(2011)}]{Sullivan2011}
Sullivan, John. 2011.
\newblock \enquote{20 {{Reasons Why Weak Managers Never Hire A-level Talent}}.}
\newblock Accessed 31 October 2025. \url{https://drjohnsullivan.com/uncategorized/20-reasons-why-weak-managers-never-hire-a-level-talent/}.

\bibitem[{{The Goldman Sachs Group, Inc.}(2023)}]{Goldman2023Q2}
{The Goldman Sachs Group, Inc.} 2023.
\newblock \enquote{Second Quarter 2023 Earnings Results.}
\newblock U.S. Securities and Exchange Commission, EDGAR database.
\newblock Accessed October 19, 2025. \url{https://www.sec.gov/Archives/edgar/data/886982/000119312523189212/d523443dex991.htm}.

\bibitem[{Tullock(1980)}]{Tullock1980}
Tullock, Gordon. 1980.
\newblock \enquote{Efficient Rent Seeking.}
\newblock In \emph{Toward a {{Theory}} of the {{Rent-seeking Society}}}. College Station: Texas A \& M University, 97--112.

\bibitem[{Wu and Liu(2026)}]{WuLiu2026}
Wu, Hugh~Xiaolong and Shannon~X. Liu. 2026.
\newblock \enquote{The {{Trade-Offs}} of {{Letting Local Managers Make Hiring Decisions}}.}
\newblock \emph{Management Science} .

\bibitem[{Zaman and Lakhani(2024)}]{ZamanLakhani2024}
Zaman, Hashim and Karim~R. Lakhani. 2024.
\newblock \enquote{Determinants of {{Top-Down Sabotage}}.}
\newblock Available at SSRN: \url{https://papers.ssrn.com/abstract=4941749}.

\end{thebibliography}

\end{document}